\documentclass[a4paper,11pt]{article}

\usepackage{relsize}
\usepackage{epigraph}
\usepackage{nicefrac}
\usepackage[margin=1in]{geometry}
\usepackage{amsmath,amssymb,amsthm}
\usepackage{hyperref}
\usepackage{fullpage}
\hypersetup{colorlinks=true,linkcolor=black,citecolor=blue}

\usepackage{amsmath, amsthm, amssymb, amstext, comment, graphicx, fullpage}

\usepackage{mathtools}
\usepackage{float}
\usepackage[titletoc,title]{appendix}
\usepackage{graphicx, color}
\usepackage[classfont=sanserif,langfont=roman,funcfont=italic]{complexity}
\usepackage{caption}
\usepackage{subcaption}
\usepackage{environ}
\NewEnviron{box1}[1]{%
	\begin{center}\fbox{\parbox{6in}{%
				{\centering\scshape #1\par}%
				\parskip=1ex
				\everypar{\hangindent=1em}%
				\BODY
}}\end{center}}

\newtheorem{theorem}{Theorem}[section]

\newtheorem{definition}{Definition}

\newtheorem{corollary}{Corollary}

\newtheorem{lemma}{Lemma}

\newtheorem{open}{Open}
\newtheorem{example}{Example}

\usepackage{booktabs} 

\renewcommand{\vec}[1]{\mathbf{#1}}

\usepackage{comment}
\newcommand{\floor}[1]{\lfloor #1 \rfloor}

\newcount\Comments
\definecolor{darkgreen}{rgb}{0,0.7,0}
\newcommand{\kibitz}[2]{\ifnum\Comments=1\textcolor{#1}{#2}\fi}

\newtheorem{observation}{Observation}

\title{Walrasian Dynamics in Multi-unit Markets \footnote{This project has received funding from the European Research Council (ERC) under the European Unions Horizon 2020 research and innovation programme (grant agreement No 740282), and from the ISF grant 1435/14 administered by the Israeli Academy of Sciences and Israel-USA Bi-national Science Foundation (BSF) grant 2014389.
Aris Filos-Ratsikas was supported by the ERC Advanced Grant 321171 (ALGAME)}}
\author{
	Simina Br\^anzei\footnote{Purdue University, USA. E-mail: \textcolor{blue}{\href{mailto:simina.branzei@gmail.com}{simina.branzei@gmail.com}}.}\\
	\newline
	\and
	Aris Filos-Ratsikas\footnote{Oxford University, United Kingdom. E-mail: 
		\textcolor{blue}{\href{mailto:aris.filos-ratsikas@cs.ox.ac.uk}{aris.filos-ratsikas@cs.ox.ac.uk}}.}
}

\date{}

\begin{document}
	
	\maketitle
\begin{abstract}
	
In a multi-unit market, a seller brings multiple units of a good and tries to sell them to a set of buyers that have monetary endowments. While 
a Walrasian equilibrium does not always exist in this model, natural relaxations of the concept that retain its desirable fairness properties do exist.

We study the dynamics of (Walrasian) envy-free pricing mechanisms in this environment, showing that for any such pricing mechanism, the best response dynamic starting from truth-telling converges to a pure Nash equilibrium with small loss in revenue and welfare. Moreover, we generalize these bounds to capture all the Nash equilibria for a large class of (monotone) pricing mechanisms. We also identify a natural mechanism, which selects the minimum Walrasian envy-free price, in which for $n=2$ buyers the best response dynamic converges from any starting profile, and for which we conjecture convergence for any number of buyers.

\end{abstract}

\maketitle
\section{Introduction}

The question of allocating scarce resources among participants with heterogeneous preferences is one of the most important problems faced throughout the history of human society, starting from millenary versions---such as the division of land---to modern variants---such as allocating computational resources to the users of an organization or selling goods in eBay auctions.
A crucial development in answer to this question occurred with the invention of money, i.e. of pricing mechanisms that aggregate information about the supply and demand of the goods in order to facilitate trade. The price mechanism was formalized and studied systematically starting with the 19th century, in the works of Fisher \cite{BS00}, \cite{Walras74} and \cite{AD54}.

The basic setting is that of a set of participants that come to the market with their initial endowments and aim to purchase goods in a way that maximizes their utility subject to the initial budget constraints. Walrasian equilibria are outcomes where demand and supply meet, have been shown to exist under mild conditions when the goods are perfectly divisible, and satisfy very desirable efficiency properties (see, e.g., \cite{AD54, AGT_book}).

The beautiful general equilibrium theory rests on several idealized conditions that are not always met. In particular, the real world is fraught with instances of allocating indivisible goods, such as a house or a piece of jewelry. Walrasian equilibria can disappear in such scenarios, and the classes of utilities for which they continue to exist are very small \cite{KC82,GS99}, requiring among other things that buyers have essentially unbounded budgets of intrinsic value to them (i.e. monetary endowments so large as to always ``cover'' their valuations for the other goods in the market). However this is a strong assumption that violates basic models such as the Fisher and exchange market, where bounded endowments are an essential element of the profiles of the agents \cite{dobzinski2012multi,AD54}.

Nevertheless, we would like to retain at least partially the crisp predictions of Walrasian equilibrium theory in the face of indivisibilities, such as the fairness properties\footnote{One such fairness property is envy-freeness.} of the pricing and the fact that it allows each buyer to freely purchase their favorite bundle at those prices. One method that has been proposed to accomplish this is relaxing the clearing requirement; for instance, in the case of one seller that brings multiple goods, it can be acceptable that at the end of the trade there are some leftover items in stock. The crucial optimality condition of Walrasian equilibria can be maintained: each buyer purchases an optimal bundle at the current prices. This notion of equilibrium is known as ``(Walrasian) envy-free pricing'' and refers to the fact that no buyer should ``envy'' any bundle that it could afford in the market \cite{guruswami2005profit}. Envy-free pricing is guaranteed to exist in very general models and the question becomes to compute one that recovers some of the efficiency properties of exact Walrasian equilibria.

A high level scenario motivating our analysis is an online marketplace where a seller posts goods for sale, while interested buyers compete for obtaining them by submitting bids that signal their interest in the goods. The buyers have time to repeatedly update their bids after seeing the bids submitted by others, in order to
get better allocations for themselves or potentially lower the price.
This process continues until a stable state is reached, point at which the seller allocates the goods based on the final bids.
We are interested in understanding the outcomes at the stable states of this dynamic in multi-unit markets, and in particular in quantifying the social welfare of the participants and revenue extracted by the seller. 

Multi-unit auctions have been advocated as a lens to the entire field of (algorithmic) mechanism design \cite{Nisan14} and our scenario is in fact
\emph{an environment where repeated best-response provides reasonable approximation guarantees}'',
stated explicitly as a future research question by Nisan et al. \cite{nisan2011best}.

\subsection{Our Results}

We study the (Walrasian) envy-free pricing problem in one of the most basic scenarios possible, namely linear multi-unit markets with budgets. There is one seller who comes equipped with $m$ units of some good (e.g. chairs), while the buyers bring their budgets. The seller has no value for the items, while the buyers value both their money and the goods. 

\smallskip

The buyers are strategic and a mechanism for envy-free pricing will have to elicit the valuations from the buyers. The budgets are known. In the economics literature, budgets are viewed as hard information (quantitative), as opposed to preferences, which represent soft information and are more difficult to verify (see, e.g., \cite{Petersen}). Bulow, Levin and Milgrom \cite{bulow2009winning} provide an example of such an inference in a real-life actions. Known budgets are also studied in the auctions literature \cite{dobzinski2012multi,fiat2011single,laffont1996optimal}. 

\smallskip

Our goal is to understand the best response dynamics as well as the entire Nash equilibrium set of such a mechanism, together with the social welfare and revenue attained at the stable states. For the best response dynamic, we assume that at each step of the process a buyer observes the bids of others, then updates its own bid optimally given the current state of the market. For the Nash equilibrium set, we will consider equilibria in which the buyers do not overbid.\footnote{We discuss and motivate this assumption in detail in Section \ref{sec:nashset}.}
Our results are for (Walrasian) envy-free pricing mechanisms and can be summarized as follows.

\begin{theorem}[Convergence] \label{thm:maininformally}
	Let $\textrm{A}$ be any mechanism. Then the best response dynamic starting from the truth-telling profile converges to a pure Nash equilibrium of $\textrm{A}$. Compared to the truth-telling outcome, the equilibrium reached has the property that 
	\begin{itemize}
		\item the utility of each buyer is (weakly) higher, and 
		\item the number of units received by each buyer is (weakly) larger, possibly with the exception of the last deviator.
	\end{itemize} 
	The best response may cycle when starting from a non-truthful profile.
\end{theorem}

Our next theorem shows that convergence can be slow in the worst case. Note the input to the mechanism (i.e. valuations and budgets) are drawn from a discretized domain (i.e. a grid), while the price is selected from an output domain that is also discretized.\footnote{Note that if the input and output domain are continuous, the best response may not be well defined; see Example \ref{ex:one}.}

\begin{theorem}[Lower bound on convergence time] (informal)
	There exists a mechanism such that for arbitrary (but fixed) market parameters, the best response dynamic takes $\Omega(1/\epsilon)$ steps to converge, where $\epsilon > 0$ is the maximum distance between consecutive entries in the input and output domains.
\end{theorem}

\noindent On the positive side, we identify a natural class of ``consistent'' mechanisms, containing welfare and revenue maximizing (and approximating) ones, for which convergence is much faster.

\begin{theorem}[Convergence time of consistent mechanisms]
	For any consistent mechanism $\textrm{A}$, the best response dynamic starting from the truth-telling profile converges to a pure Nash equilibrium of $\textrm{A}$ in at most $n$ steps.
\end{theorem}

\medskip

\noindent We also quantify the social welfare and revenue attained at the end of the best response process. Some of our bounds use the notion of budget share, which captures the level of competition in the market by quantifying the effect that the budget of a single buyer can have on the revenue. The more competitive the market is, the smaller the budget share (for similar notions of competitiveness, see \cite{cole2015does,dobzinski2012multi}).

\begin{theorem}[Revenue and welfare of the dynamic starting from truth-telling]
	Let $A$ be any mechanism. Then the best response dynamic starting from truth-telling has the property that the Nash equilibrium reached, compared to the truth-telling state, is such that:
	\begin{itemize}
		\item the (additive) loss in social welfare (compared to the maximum possible) is at most the maximum budget.
		\item the revenue is approximated within a (multiplicative) factor of at least $\left(\beta-\alpha\right)/2$, where $ \beta \in [0,1]$ is the approximation ratio of mechanism $A$ for the optimal revenue and $\alpha$ is the budget share of the market.
	\end{itemize}
\end{theorem}

\noindent We show that similar approximation guarantees hold for all the (non-overbidding) Nash equilibria of monotone mechanisms.

\begin{theorem}[Revenue and welfare in any non-overbidding Nash equilibrium] \label{thm:all_equilibria_revsw}
	Let $\textrm{A}$ be a monotone mechanism. Then in any pure Nash equilibrium of $\textrm{A}$ where buyers do not overbid, compared to the truth-telling outcome of $A$:
	\begin{itemize}
		\item the (additive) loss in welfare is at most $\gamma_A \cdot B^*$, where $B^*$ is the maximum budget
		\item the revenue is approximated within a multiplicative factor of at least $\left(\beta-\gamma_A \cdot \alpha\right)/2$, 

where $\beta \in [0,1]$ is the approximation ratio of the mechanism $A$ for the optimal revenue and $\alpha$ is the budget share of the market
	\end{itemize}
	such that $\gamma_A$ is the maximum number of buyers that can receive partial allocations by $\mathrm{A}$.
\end{theorem}

The notion of ``partial'' allocations refers to buyers that are indifferent between multiple bundles.\footnote{In this model, the buyers that are indifferent among multiple bundles are precisely those whose valuation exactly matches the price, and their demand set contains all the bundle sizes that they can afford at that price. Note the buyers whose valuation is strictly higher than the price want the maximum number of units they can afford, while the buyers whose valuation is strictly lower than the price want zero units, so for these two latter types the allocations are never ``partial''.} A partial allocation for a buyer is a bundle that is neither the largest in its demand nor the smallest. In principle if a mechanism gives to many buyers bundles that are neither the largest nor the smallest desired by that buyer, the revenue and welfare of that mechanism can vanish in the equilibrium. We identify natural mechanisms, such as with greedy allocation rules (which take the indifferent buyers in some order and allocate each of them fully before moving to the next one), for which the number of buyers receiving partial allocations is at most $1$.

\smallskip
Finally, we study the \textsc{All-or-Nothing} mechanism, which selects the minimum (Walrasian) envy-free price and allocates the semi-hungry buyers either all the units they can afford or none. 
\begin{theorem} \label{thm:aon_awesome2}
	For $n=2$ buyers, the best-response dynamic starting from \emph{any profile} of the \textsc{All-Or-Nothing} mechanism converges to a pure Nash equilibrium. 
\end{theorem}
\begin{figure}[h!]
	\centering
\begin{subfigure}{0.5\textwidth}
		\includegraphics[width=1\textwidth, trim={0 0 0 0},clip]{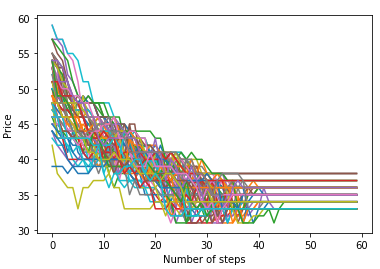}
	\end{subfigure}

		\caption{The convergence properties of the \textsc{All-Or-Nothing} mechanism when the dynamic starts from 100 different arbitrary profiles, for $n=25, m=20$. The budgets are drawn uniformly at random from 
		 $[1,125]$. 
		The differently coloured lines indicate different starting profiles $\mathbf{s}^{(0)}$.} \label{figfig}
\end{figure}

This mechanism was shown to be both truthful and to simultaneously approximate the optimal revenue and welfare revenue within constant factors (at the truthful state) in previous work \cite{BFMZ16}. 
Using our theorems above, since the mechanism is monotonic, we get the following performance guarantees for \textsc{All-or-Nothing}:

\begin{corollary} \label{cor:aon_awesome}
In every non-overbidding equilibrium of \textsc{All-Or-Nothing}, for any number of buyers,
\begin{itemize}
	\item [-]the (additive) loss in welfare is at most the maximum budget,
	\item [-]the revenue is approximated within a (multiplicative) factor of at least $\min \left\{0.25\right.$, $\left.(1-\alpha)/2\right\} - \alpha/2$ where  $\alpha$ is the budget share of the market.
\end{itemize}
In particular, in competitive markets (where $\alpha$ approaches $0$), mechanism \textsc{All-Or-Nothing} approximates the optimal revenue within a factor of $4$ in any non-overbidding equilibrium.
\end{corollary}

\subsection{Related work}

The notion of the Walrasian equilibrium was formulated and studied systematically as early as the beginning of the 19th century, with the foundational works of Fisher \cite{BS00}, \cite{Walras74}, and \cite{AD54}. Its simple pricing scheme and it superior fairness properties, have lead to the employement of the Walrasian equilibrium as an auction mechanism for selling goods \cite{guruswami2005profit,babaioff2014efficiency}. The 2000 SMRA auction for 3G licences to UK providers (coined as ``The Biggest Auction Ever'' in \cite{binmore2002biggest}) as well as the Product-Mix auction \cite{klemperer2010product} run by the Bank of England are prime real-life examples of the employment of Walrasian pricing mechanisms.

Mechanisms for Walrasian pricing have been studied in the literature for various settings, with the focus being mainly on computational issues, i.e. designing polynomial time mechanisms to achieve good approximations for the revenue and the social welfare objectives \cite{feldman2012revenue,monaco2015revenue,cheung2008approximation,balcan2008item}. Crucially, most of these works do not take the incentives of buyers into account. Our results provide worst-case guarantees for all mechanisms for the setting of linear multi-unit auctions with budgets.

The multi-unit model is central in combinatorial auctions and has been studied in a large body of literature \cite{dobzinski2007mechanisms,dobzinski2012multi,feldman2012revenue,dobzinski2015multi,BCIMS05,bartal2003incentive,GMP13,DHP17}, in particular for the case of additive valuations \cite{dobzinski2012multi, feldman2012revenue}.
The literature on best-response dynamics on games is rich \cite{roughgarden2009intrinsic,chien2007convergence,awerbuch2008fast}, with questions raised about convergence and convergence time and the properties of the equilibrium which is reached as the results of the dynamic. A line of work has also considered best-response dynamics in auctions \cite{dutting2017best,nisan2011best,cary2008best,BBN17,CD11,CDEG+14}; this is relevant to this work as pricing mechanisms in our setting can be seen as simple and fair auction mechanisms for selling indivisible goods. 

Our work is also related to the literature on understanding dynamics of auctions under behavioral assumptions such as no-regret learning and equilibria of auctions (e.g. for Nash, correlated, and coarse-correlated equilibria) with corresponding price of anarchy bounds; see, e.g. the work of  \cite{RST17,DS16,MT12,HKMN11,bhawalkar2011welfare,lucier2010price}.
One could view our results in Section \ref{sec:nashdynamic} as \emph{dynamic Price of Anarchy} results \cite{BCMP13} and the results in Section \ref{sec:nashset} as \emph{undominated Price of Anarchy} results \cite{caragiannis2015bounding}. Our setting is also related to an ascending auction proposed by Ausubel \cite{Ausubel04}, in which the price keeps increasing while the buyers reduce their demands accordingly, with the feature that whenever the total demand of all except one of the bidders (say $i$) drops below the supply, the remaining items are allocated to bidder $i$ at the current price. Unlike Ausubel's auction, we study auctions in which the items (which are identical units) are allocated at the same price to everyone.

Contrary to some of the related work on dynamics (e.g. see \cite{Kesselheim16}), we do not need to make any assumptions about the order in which the buyers best-respond; our results are independent of this order. A summary on dynamics for deciding the allocation of public goods can be found in \cite{Laffont87}. Best response dynamics starting from the true profile have been studied before in voting (see, e.g. \cite{BCMP13,MPRJ10}). We emphasize here that while a truthful mechanism with good welfare and revenue guarantees exists for our setting  \cite{BFMZ16}, providing guarantees for very general classes of mechanisms as we do here is quite important, as the mechanisms that might actually be employed in reality might not be truthful. A prominent example of this phenomenon is the Generalized Second Price Auction, which is actually used in practice in favour of its famous truthful counter-part, the VCG mechanism (see \cite{AGT_book}) and in fact, understanding such non-truthful mechanisms theoretically has been a focus of the related literature \cite{caragiannis2015bounding,de2013inefficiency}


Finally, the question of understanding what can be achieved with item pricing is particularly important  given the recent emphasis on the theme of complexity versus simplicity in mechanism design \cite{hartline2009simple}, which is inspired by the computational perspective. Item pricing is a qualitative notion of simple pricing (as opposed to more complex pricing schemes, such as pricing different bundles arbitrarily), and is frequently used for selling goods in supermarkets. More formal definitions of complexity include communication and time, while recently notions of simplicity have been proposed specifically for economic environments, such as the menu-size of auctions \cite{HN12},
or simplicity in the sense that participants can easily reason about the mechanism's properties, studied using a framework of verifiable mechanisms \cite{BP15} and obvious strategy-proofness \cite{Li17}.

\section{Model and Preliminaries}\label{sec:model}

A linear \emph{multi-unit market} is composed of a set $N = \{1, \ldots, n\}$ of buyers and one seller that brings $m$ (indivisible) units of a good.
Each buyer $i$ has a budget\footnote{Monetary endowment of intrinsic value to all the buyers and the seller.} $B_i > 0$, and a valuation $v_i$ per unit, such that the value of buyer $i$ for $k$ units is $v_{i} \cdot k$.
The seller has no value for the units and its goal is to extract money from the buyers, while the buyers value both the money and the good and their goal is to maximize their utility by purchasing units of the good as long as the transaction is profitable.


The seller will set a price $p$ per unit, such that any buyer can purchase $k$ units at a price of $k \cdot p$.
The outcome of the market will be a pair $(\vec{x},p)$, where $p$ is the unit price and $\vec{x} = (x_1, \ldots, x_n) \in \mathbb{Z}^{n}_{+}$ is an \emph{allocation}, with the property that $x_i$ is the number of units received
by buyer $i$ and  $\sum_{i=1}^{n} x_i \leq m$.

\smallskip

The \emph{utility} of buyer $i$ for receiving $k$ units at price $p$ is:
\begin{equation*}
u_i(p, k)=\begin{cases}
v_{i} \cdot k - p \cdot k, & \text{if $p \cdot k \leq B_i$}\\
- \infty, & \text{otherwise}
\end{cases}
\end{equation*}
The \emph{demand} of a buyer $i$ at price $p$ is defined as the set of optimal bundles for the buyer given the price and its budget. Since the units are indistinguishable, the demand set will simply contain all the bundle sizes that the buyer would maximally prefer:
\begin{equation*}
D_i(p)=\begin{cases}
\left\{\min\left\{\lfloor \frac{B_i}{p} \rfloor,m\right\} \right\}, & \text{if $v_i > p$}\\[4pt]
\left\{0, 1, \ldots,\min\left\{\lfloor\frac{B_i}{p}\rfloor ,m\right\} \right\},  & \text{if $v_i = p$}\\[4pt]
\ \ \{0\}, & \text{if $v_i < p$}.
\end{cases}
\end{equation*}
Thus, if the valuation is higher than the price per unit, the buyer demands as many units as it can afford (up to exhausting all the units), while if the valuation is lower, the buyer wants zero units.

A buyer $i$ is said to be \emph{hungry} at a price $p$ if $v_i > p$, \emph{semi-hungry} if $v_i=p$, \emph{interested} if $v_i \geq p$, and \emph{uninterested} if $v_i < p$. 
\medskip

\noindent \textbf{Walrasian (Envy-Free) Pricing:} An allocation and price $(\vec{x},p)$ represent an \emph{envy-free pricing} \footnote{We note that more complex forms of envy-free pricing are possible, such as setting a different price per bundle, but our focus will be on the simplest type of unit pricing.} if the price per unit is $p$ and $x_i \in D_i(p)$ for all $i \in N$, i.e. each buyer gets a number of units in its demand set at $p$. A price $p$ is an \emph{envy-free price} if there exists an allocation $\vec{x}$ such that $(\vec{x},p)$ is an envy-free pricing. An allocation $\vec{x}$ \emph{can be supported} at price $p$ if there is an envy-free pricing $(\vec{x},p)$.

While an envy-free pricing can always be obtained (e.g. set $p = 1 + B_1 + \ldots + B_n$), it is not necessarily possible to sell \emph{all} the units in an envy-free way as can be seen from the next example.

\begin{example}[Non-existence of envy-free clearing prices]
	Consider a market with $m=3$ units and two buyers, $Alice$ and $Bob$, with valuations $v_{Alice} = v_{Bob} = 1.1$ and budgets $B_{Alice} = B_{Bob} = 1$. At any price $p > 0.5$, no more than $2$ units can be sold in total due to budget constraints. At $p \leq 0.5$, both Alice and Bob are interested and want at least $2$ units each, but there are only $3$ units in total.
\end{example}



\noindent \textbf{Mechanisms and Objectives:} We study envy-free pricing mechanisms and will be interested in the social welfare of the buyers and the revenue of the seller. An \emph{mechanism $\textrm{A}$ for envy-free pricing} receives as input a market $\mathcal{M} = (\vec{v,B},m)$ and outputs an envy-free pricing $(\vec{x},p)$. The valuations are \emph{private} and will be elicited by $A$ from the agents. The budgets and number of units are known. 

The \emph{social welfare} at an envy-free pricing $(\vec{x}, p)$ is the total value of the buyers for the goods allocated to them, while the \emph{revenue} is the amount of money received by the seller, respectively:
$$\mathcal{SW}(\vec{x}, p) = \sum_{i=1}^n v_{i}\cdot x_i, \ \ \ \ \ \ \mathcal{REV}(\vec{x}, p) = \sum_{i=1}^n x_i \cdot p.$$

\noindent\textbf{Incentives:} Buyers are rational agents who strategize when reporting their valuations, in order to gain better allocations at a lower price. The truthful valuation profile will be denoted by $\vec{v}= (v_1, \ldots, v_n)$. However, each buyer $i$ can report any number $v_i' \in \Re_{+}$ as a value, and this will be its \emph{strategy}.

Given a mechanism $A$ and market $\mathcal{M}= (\vec{v}, \vec{B}, m)$, a strategy profile $\vec{s} = (s_1, \ldots, s_n)$ is a \emph{pure Nash equilibrium} (or simply an equilibrium) if no buyer $i$ can find an alternative strategy $s_i'$ that would improve its utility under mechanism $A$, given that the strategies of the other buyers are fixed.\\

\noindent\textbf{Discrete Domain:} The valuations and the budgets come from a discrete domain (an infinite grid): $\mathbb{V} = \{k\cdot \epsilon \; | \; k \in \mathbb{N}\}$ for some $\epsilon > 0$. An mechanism will be allowed to choose the price from an output grid $\mathbb{W} = \{k\cdot \delta \; | \; k \in \mathbb{N}\}$, for some $\delta > 0$.  

If the input and output domains are continuous, the best-response may not be well-defined.

\begin{example}[Best-response dynamic]\label{ex:one}
	Consider a market with $m=2$ units and two buyers, Alice and Bob, with valuations $v_{Alice} = v_{Bob} = 2$ and budgets $B_{Alice} = B_{Bob} = 2$. Consider a mechanism $A$ that selects the second-highest valuation as the price and allocates the semi-hungry buyers by taking them in lexicographic order and giving each as many units as possible before moving to the next one. Then the price set by $A$ is $p = 2$, with allocation $x_{Alice} = 2$ and $x_{Bob} = 0$. 
	
	If Alice deviates to $s_{Alice}=1$, then Alice receives $0$ units, since Bob will demand $2$ units. At any report $s_{Alice} >1$, Alice receives $1$ unit for a strictly positive utility and therefore to minimize the price Alice would have to find the minimum number in the open set $(1,\infty)$. Finally, any report $s_{Alice} > 2$ would result in Alice having negative utility, while a report $s_{Alice} < 1$ would result in her getting no units for a utility of zero, and so no improvement compared to the truth telling outcome.
	
	On the other hand, on the discrete domain, Alice has a best response by deviating to reported valuation $1+\epsilon$, case in which the profile $(1+\epsilon,2)$ is a pure Nash equilibrium.
\end{example}

\section{Dynamics}\label{sec:nashdynamic}

In this section, we consider best-response dynamics and will focus on the truth-telling profile as a natural starting point of the dynamic. Given a market $\mathcal{M} = (\vec{v}, \vec{B}, m)$, we will denote by $\mathcal{H}_p(\vec{v})$, $\mathcal{S}_p(\vec{v})$, and $\mathcal{I}_p(\vec{v})$ the hungry, semi-hungry, and interested buyers at $p$, respectively.
For a strategy profile $\vec{s}$, $\vec{s_{-i}}$ will denote the strategy profile of all other buyers except the $i$'th one.
The omitted proofs from this section can be found in Appendix \ref{app:nashdynamic}.

\medskip

Our main theorem in this section is that the best response dynamic always converges to a pure Nash equilibrium when starting from the truth telling profile.

\begin{theorem}[Convergence]\label{thm:convergence}
	Let $\textrm{A}$ be any mechanism. Then the best response dynamic starting from the truth-telling profile converges to a pure Nash equilibrium of $\textrm{A}$. Compared to the truth-telling outcome, the Nash equilibrium reached has the property that
	\begin{itemize}
		\item the utility of each buyer is (weakly) higher, and 
		\item the number of units received by each buyer is (weakly) larger, possibly with the exception of the last deviator.
	\end{itemize} 
\end{theorem}
\begin{proof}
	For any step $k$ in the best response process, let $\vec{s}^k = (s_1^k, \ldots, s_n^k)$ be the vector of valuations reported by all the buyers at $k$, and $p^k$ the price output by $\textrm{A}$ in this round. For $k=0$ we have $\vec{s}^0 = \vec{v}$ as the initial (true) valuations.
	We show by induction the following properties:
	\begin{enumerate}
		\item The price strictly decreases throughout the best response process, i.e. $p^{k-1} > p^{k}$ for every step $k \geq 0$.
		\item The buyers that appear hungry at $p^k$, i.e. $\mathcal{H}_{p^k}(\vec{s}^k)$, are also hungry at this price with respect to their true valuations, i.e. $v_i > p^k$ for all $i \in \mathcal{H}_{p^k}(\vec{s}^k)$. The buyers that appear semi-hungry at $p^k$, i.e. $\mathcal{S}_{p^k}(\vec{s}^k)$, have not changed their inputs, possibly with the exception of the last deviator, whose true valuation is greater or equal to $p^k$. The buyers that appear uninterested at $p^k$ are honest, i.e. $N\setminus \mathcal{I}_{p^k}(\vec{s}^k)$ is such that $s_i^k = v_i$ for all $i \in \mathcal{I}_{p^k}(\vec{s}^k)$.
	\end{enumerate}
	At $k=0$ property $1)$ holds trivially since there is no previous price (we can define $p^{-1} = \infty$), while $2)$ holds for the truth-telling valuations. Suppose properties $(1-2)$ hold for all steps up to $k-1$. We show they also hold at step $k$. Let $i$ be the buyer that deviates in step $k$ to some valuation $s_i^k$. Let $v_j^k = v_j^{k-1}$ for all buyers $j \neq i$ and $p^k$ the new price selected by $\textrm{A}$ on input $\vec{v}^k$. We consider a few cases depending on the deviating buyer.\\
	
	\noindent \emph{Case 1:} \emph{Buyer $i$ appears hungry in round $k-1$: $s_i^{k-1} > p^{k-1}$}. By the induction hypothesis, $i \in \mathcal{H}_{p^{k-1}}(\vec{s}^{k-1})$, and so $v_i > p^{k-1}$. Then buyer $i$ receives a maximal allocation in round $k-1$ given the price. If $i$ increases the price in round $k$, i.e. $p^k > p^{k-1}$, this can only result in buyer $i$ getting (weakly) fewer units at a higher price, which worsens $i$'s utility compared to round $k-1$. This cannot be a best response. Similarly, if $i$'s deviation resulted in the same price for round $k$, $p^k = p^{k-1}$, buyer $i$ would receive at most as many units at the same price, which is also not an improvement. Thus it must be the case that $p^k < p^{k-1}$. 
	
	Moreover, buyer $i$'s new valuation cannot be smaller than $p^k$ since that would result in $i$ getting no units in round $k$, which cannot be better than round $k-1$ where buyer $i$ was hungry with respect to its true value. Since $v_i > p^{k-1} > p^k$, we get $s_i^k \geq p^k$ and $v_i > p^k$.
	The set $\mathcal{H}_{p^k}(\vec{s}^k)$ of buyers that appear hungry in round $k$ contains, in addition to possibly buyer $i$, also
	\begin{itemize}
		\item the buyers in $\mathcal{H}_{p^{k-1}}(\vec{s}^{k-1})\setminus \{i\}$, which remain hungry with respect to their true valuations in round $k$ since the price decreased compared to round $k-1$ while their reports have not changed,
		\item the buyers with valuations $s_i^{k-1} \in (p^k, p^{k-1})$, which are honest in round $k-1$ and remain honest in round $k$, with true values strictly above $p^k$, since none of them deviated in between.
	\end{itemize}
	The set $\mathcal{S}_{p^k}(\vec{s}^k)$ of buyers that appear semi-hungry in round $k$ contains only buyers that were honest in round $k-1$, possibly with the exception of $i$. Thus all the buyers in $\mathcal{S}_{p^k}(\vec{s}^k) \setminus \{i\}$ have their true (and reported) valuations equal to $p^k$. If buyer $i$ deviates to $s_i^k = p^k$, then by previous arguments, $v_i > p^k$. Finally, the buyers $N\setminus \mathcal{I}_{p^k}(\vec{s}^k)$ that appear uninterested in round $k$ were honest in round $k-1$ and remain so in round $k$, since the deviator $i$ is not one of them. Thus properties \emph{(1-2)} hold in round $k$.
	\medskip
	
	\noindent \emph{Case 2:} \emph{Buyer $i$ appears semi-hungry in round $k-1$: $s_i^{k-1} = p^{k-1}$}. By the description of $\mathcal{S}_{p^{k-1}}(\vec{s}^{k-1})$, if buyer $i$ is not honest in round $k-1$, it must be the case that $i$ was the deviator in round $k-1$. If buyer $i$ manages to strictly increase its utility from $k-1$ to $k$, then $i$'s input in round $k-1$ was not a best response to round $k-2$, which is a contradiction. By the induction hypothesis, it follows that buyer $i$ is honest in round $k-1$. Then $i$ cannot wish to keep the price the same or increase it, since that would go above its true valuation, which is $v_i = s_i^{k-1}$. Thus $i$ can only decrease the price, so $p^k < p^{k-1}$ and $v_i > p^k$. Similarly, for $s_i^k$ to be an improvement, buyer $i$ should appear interested on the new instance, so $s_i^k \geq p^k$. 
	
	The set of buyers $\mathcal{H}_{p^k}(\vec{s}^k)$, which appear hungry in round $k$, is a superset of $\mathcal{H}_{p^{k-1}}(\vec{s}^{k-1})$ and contains all the buyers with valuations $s_i^{k-1}$ in $(p^k, p^{k-1})$, as well as $i$ in case $s_i^k > p^k$. The set of semi-hungry buyers in round $k$, $\mathcal{S}_{p^k}(\vec{s}^k)$, contains only buyers that were honest in round $k-1$, possibly with the exception of buyer $i$ in case $s_i^k = p^k$. Finally, the uninterested buyers in round $k$, $N \setminus \mathcal{I}_{p^k}(\vec{s}^k)$, form a subset of $N\setminus\mathcal{I}_{p^{k-1}}(\vec{s}^{k-1})$, were honest in round $k-1$, and have not changed their inputs in between. Thus the sets in round $k$ satisfy the required properties and the price decreased compared to round $k-1$.
	\medskip
	
	\noindent \emph{Case 3:} \emph{Buyer $i$ appears uninterested in round $k-1$: $s_i^{k-1} < p^{k-1}$}. Then $i \in N\setminus \mathcal{I}_{p^{k-1}}(\vec{s}^{k-1})$, which by the induction hypothesis only contains honest buyers. Then buyer $i$'s utility can only improve by decreasing the price and appearing hungry or semi-hungry in the next round, so $p^k < p^{k-1}$, $s_i^k \geq p^k$, and $v_i > p^k$. Similarly to the previous cases, $\mathcal{H}_{p^k}(\vec{s}^k)$ contains all of $\mathcal{H}_{p^{k+1}}(\vec{s}^{k+1})$ (which continue to be hungry in round $k$ with respect to both their true and reported valuations), the buyers with valuations $s_i^{k-1} \in (p^k, p^{k-1})$ (which were honest in round $k-1$ and have not changed in between), and buyer $i$ in case $s_i^k > p^k$ (recall $v_i > p^k$). The set $\mathcal{S}_{p^k}(\vec{s}^k)$ contains buyers which were honest in round $k-1$ and have the same reports in round $k$, possibly with the exception of buyer $i$ in case $s_i^k = p^k$. The set $N \setminus \mathcal{I}_{p^k}(\vec{s}^k)$ is a subset of $N \setminus \mathcal{I}_{p^{k-1}}(\vec{s}^{k-1})$, all of which were honest at $k-1$ and kept the same valuations at $k$.
	
	In all three cases, we obtain that properties $(1-2)$ are maintained in round $k$, which completes the proof by induction. Then the price strictly decreases in every iteration. Since the price output is chosen from a discrete grid of non-negative entries, the best response process either stops or reaches the smallest grid point available above zero. At this point the buyers cannot decrease the price further, which implies there are no more best responses. 
	
	The improvement properties follow from the description of the sets $\mathcal{H}_{p^k}(\vec{s}^k), \mathcal{S}_{p^k}(\vec{s}^k), N \setminus \mathcal{I}_{p^k}(\vec{s}^k)$ in the round $k$ where the best response process stops, which completes the argument.
\end{proof}

The loss in the allocation of the last deviator is sometimes inevitable. 

\begin{theorem}[Necessary allocation loss]\label{thm:necessaryloss}
	There exist markets where the best response dynamic of a mechanism converges to an allocation where some buyer loses units compared to the truth-telling outcome.
\end{theorem}

In the worst case, convergence can be slow: the buyers can force a mechanism to traverse the entire input grid by slowly decreasing their reported valuations with each iteration.

\begin{theorem}[Lower bound on convergence time]
	There is a mechanism such that for arbitrary (fixed)  market parameters $\mathcal{M} = (\vec{v}, \vec{B}, m)$, the best response dynamic takes $\Omega(1/\epsilon)$ steps to converge, where the distance between consecutive entries in both input and output domain is at most $\epsilon > 0$ \footnote{E.g., the entries are at most $\epsilon$-apart when the input and output domain are discretized to have the form $\{k \cdot \epsilon \; | \; k \in \mathbb{N}\}$.}.
\end{theorem}
\begin{proof}
	Let $\epsilon > 0$. Consider the following mechanism, defined on an input grid with step size $\epsilon$: 
	
	\begin{box1}{\emph{\textbf{\textsc{Almost -- Top Mechanism:}}}}
		\emph{Input}: market with valuations $v_1 \ldots v_n$ and budgets $B_1 \ldots B_n$. The valuations are drawn from an $\epsilon$-grid, i.e. $v_i = k_i \cdot \epsilon$, where $k_i \in \mathbb{N}$ for all $i$.
		\begin{itemize}
			\item Set the price to $p = v_{i_1}$, where $v_{i_1} \geq \ldots v_{i_n}$.
			\item If $v_{i_1} - \epsilon$ is an envy-free price, adjust the price to $p = v_{i_1} - \epsilon$. 
			\item Satisfy the demands of the hungry buyers at $p$, then allocate to the semi-hungry buyers using greedy tie-breaking.
		\end{itemize}
	\end{box1}
	
	\medskip
	We claim the \textsc{Almost-top} mechanism forces the buyers to traverse the entire grid in the worst case, which will give the required lower bound.
	First note the mechanism uses greedy tie-breaking for the semi-hungry buyers. Consider any valuations $\vec{v} = (v_1, \ldots, v_n) \leq (v_1', \ldots, v_n') =  \vec{v}'$, where the prices selected by the mechanism are $p$ and $p'$, respectively. W.l.o.g., $v_{i_1} \geq \ldots \geq v_{i_n}$ and $v_{j_1}' \geq \ldots \geq v_{j_n}'$. If $v_{j_1}' > v_{i_1}$, since each value is a multiple of $\epsilon$, we get $p' \geq v_{j_1}' - \epsilon \geq v_{i_1} \geq p$. Otherwise, $v_{j_1}' = v_{i_1}$. Then $p,p' \in \{v_{i_1}, v_{i_1} - \epsilon\}$. Assume by contradiction that $p > p'$, i.e. $p = v_{i_1}$ and $p' = v_{j_1}' - \epsilon$. Since $p'$ is envy-free on $\vec{v}'$ and $\vec{v} \leq \vec{v}'$, $p'$ is also envy-free on $\vec{v}$. Moreover, $p' = v_{i_1} - \epsilon$, so $\textsc{Almost-Top}$ should have chosen $p'$ on $\vec{v}$. This is a contradiction, so $p \leq p'$. Thus the $\textsc{Almost-Top}$ mechanism is monotonic.
	
	Given $0 < \epsilon < 1$, 
	consider an instance with $n$ buyers, $m = 4$ units, valuations $v_1 = 1 + \epsilon$, $v_2 = 1$, $v_3 = \ldots = v_n = \epsilon$, budgets $B_1 = B_2 =1$ and $B_3 = \ldots = B_n = \epsilon$. On this input, \textsc{Almost-top} sets the price to $p_1 = 1$, allocating one unit to each buyer. Then buyer $1$ has a best response at $v_1' = 1 - 1/N$, since on input $(v_1', v_2)$, the mechanism sets the price to $p' = 1 - 1/N$, giving at least as many units to each buyer for a lower price. Note there is no valuation between $v_1'$ and $v_1$, while any valuation lower than $v_1'$ would still result in a price of $p'$; moreover. We argue that buyer $2$ then has a best response at $v_2'' = 1 - 2/N$, since on input $(1-1/N, 1 - 2/N)$ the mechanism would set the price to $p'' = 1 - 2/N$ and give both buyers at least as many units at a lower price. Again buyer $2$ cannot move the price lower or get more units with any lower valuation at this point. Iteratively, at every step $k$, one of the buyers has a best response at the valuation $1 - k/N$. This process stops exactly at inputs $(1/N, 2/N)$, which will result in a price of $1/N$, giving each buyer exactly $2N$ units. The price cannot drop any further since the only valuation available that is lower than $1/N$ is zero, which would not change the price from $1/N$. This completes the argument.
\end{proof}

However, many natural mechanisms, such as ones that maximize or approximate welfare and revenue \cite{BFMZ16,feldman2012revenue}) converge faster, which we capture under the umbrella of consistent mechanisms.

\begin{definition}
	An mechanism $A$ is \emph{consistent} if the following holds.
	Let $\mathcal{M} = (\vec{v}, \vec{B}, m)$ be any market on which $A$ outputs a price $p$ and an allocation $\vec{x}$. Then in every market $\mathcal{M}' = (\vec{v}', \vec{B}, m)$ for which the sets of hungry and semi-hungry buyers are the same at $p$, $A$ must output $(p,\vec{x})$ as well.	
\end{definition}

For such pricing mechanisms each buyer will only best respond at most once.

\begin{theorem}[Convergence time of consistent mechanisms]\label{dyn:conv}
	Let $\textrm{A}$ be any consistent mechanism. The best response dynamic starting from the truth-telling profile converges to a pure Nash equilibrium of $\textrm{A}$ in at most $n$ steps.
\end{theorem}

Finally, the best response dynamic can cycle when starting from non-truthful profiles.

\begin{theorem}[Best-response cycles]\label{thm:bestresponsecycles}
	There exist mechanisms for which the best response dynamic can cycle when starting from non-truthful profiles.
\end{theorem}
\begin{proof}
	Consider an instance with $2$ buyers, Alice and Bob, and $2$ units with true valuations $v_{Alice} = 1$ ,$ v_{Bob} = 2$ and budgets $B_{Alice}=B_{Bob}=1.5$. Consider the following envy-free pricing mechanism $\textrm{A}$ defined for two buyers:
	\begin{box1}{} 
		\emph{Input: market $(\vec{v},\vec{B},m)$}
		\begin{itemize}
			\item[-] If $B_{Alice} = B_{Bob} = 1.5$, then set the price as follows: $p_\textrm{A}(0.1,0.3) = 0.2, p_\textrm{A}(3,0.3) = 0.9, p_\textrm{A}(3,2)=1.5$, $p_\textrm{A}(0.1,2)=0.5$ and $p_\textrm{A}(v_{Alice},v_{Bob}) = \max\{v_{Alice},v_{Bob}\} + 100$, otherwise. \emph{(Note that by this pricing scheme semi-hungry buyers never exist.)}\medskip
			
			\item[-] Otherwise, set the price to $p_\textrm{A}(v_{Alice},v_{Bob}) = \max_{i \in \{Alice,Bob\}} v_i$. Allocate maximally to the buyer with highest valuation, while giving no units to the buyer with the second highest value. If there are ties, resolve them lexicographically.
		\end{itemize}
	\end{box1}
	Consider first the strategy profile $(0.1, 0.3)$ where the price is $0.2$. At this profile, Alice receives $0$ units whereas Bob receives two units (since the demand of a buyer at a price $p$ is $\max\{\lfloor B_i/p \rfloor, m\}$. Thus Alice has a utility of zero while Bob has strictly positive utility.
	
	The best response for Alice on this strategy profile is to respond by claiming that her value is $v_{Alice}' = 3$. On the new strategy profile $(3,0.3)$, the price will get updated to $0.9$. Then Alice will receive two units for a strictly positive utility, whereas Bob will now appear uninterested, receiving zero units and getting a utility of zero.
	
	Next, Bob can best respond by claiming that his value is $2$ (which is his true valuation). One the new strategy profile $(3,2)$ the price is $1.5$ and both Alice and Bob appear to be hungry. Note, however, that since their budgets are $B_{Alice}=B_{Bob}=1.5$ and the price is $1.5$, there is enough supply to accommodate both demands and each of them receives one unit. Bob's true valuation is $2 > 1.5$ and so his utility is positive; on the other hand, Alice's true valuation is $ 1 < 1.5$ and her utility at this price is negative.
	For this reason, Alice will best-respond by claiming that her value is $0.1$. On the new strategy profile $(0.1,2)$, the price will be set to $0.5$, which is below Alice's true valuation but above her reported valuation. At this price, Alice will receive zero items and pay nothing for a utility of zero, while Bob will receive two units and get a positive utility.
	
	Finally, Bob will best-respond by claiming that his value is $0.3$, resulting in strategy profile $(0.1,0.3)$, which coincides with the original profile, where the price is $0.2$. At this price, Bob has a higher utility than before, since it still appears to be hungry and still receives two units but at a lower price.
	Thus we reach the profile that we started from, which completes the cycle.
\end{proof}

One can see that the mechanism used in the proof of Theorem \ref{thm:bestresponsecycles} has a rather strange pricing rule. We would expect that mechanisms with more natural pricing and allocation rules could converge from any starting state. To this end, we provide an example of a known natural mechanism from the literature, the \textsc{All-Or-Nothing} mechanism \cite{BFMZ16}, which (among other desirable properties), enjoys such convergence guarantees, at least for the case of $n=2$ buyers. We postpone the discussion until Section \ref{sec:aon}.

\subsection{Social Welfare and Revenue}

In this section, we will evaluate the effect of the dynamic on the social welfare and revenue obtained by envy-free pricing mechanisms. We will show that the loss in the objectives is small if the market is sufficiently competitive. 
To measure the level of competition, we define the \emph{budget share} of a market $\mathcal{M}= (\vec{v}, \vec{B}, m)$ as 
$$\alpha(\mathcal{M}) = \max_{i \in N} \frac{B_i}{\mathcal{REV}(\mathcal{M})},
$$ where $\mathcal{REV}(\mathcal{M})$ is the maximum possible revenue for valuation profile $\vec{v}$. Intuitively, the budget share measures the fraction of the total revenue that a single buyer can be responsible for; this number is small when the market is competitive. Similar notions of market competitiveness have been studied before (see, e.g., \cite{dobzinski2012multi,cole2015does}).

We start by measuring the welfare loss due to the best-response dynamics.

\begin{theorem}[Welfare of the dynamic from truth-telling] \label{thm:welfare-guarantee}
	Let $\textrm{A}$ be any mechanism. Then the best response dynamic starting from the truth-telling profile converges to a pure Nash equilibrium of $\textrm{A}$, whose loss in welfare compared to the truth-telling profile is at most the maximum budget.
\end{theorem}
\begin{proof}
	Consider a market $\mathcal{M}= (\vec{v}, \vec{B}, m)$. For any valuation profile $\tilde{v}$, let $p_\textrm{A}(\vec{\tilde{v}})$ denote the price and $x_i^\textrm{A}(\vec{\tilde{v}})$ the allocation of buyer $i$ computed by $A$ given the market $\mathcal{M}' = (\vec{\tilde{v},B},m)$. Let $\vec{s}$ be the strategy profile in the Nash equilibrium reached. By Theorem \ref{thm:convergence}, there is only a single buyer $i \in N$ for which $x_i^\textrm{A}(\vec{s}) \leq x_i^\textrm{A}(\vec{v})$ and the welfare loss is only due this buyer. From Theorem \ref{thm:convergence}, we know that on profile $\vec{s}$, the utility of buyer $i$ is at least as high as its utility on the truth-telling profile $\vec{v}$. In other words, 
	$(v_i - p_\textrm{A}(\vec{s})) \cdot x_i^\textrm{A}(\vec{s}) \geq (v_i - p_\textrm{A}(\vec{v})) \cdot x_i^\textrm{A}(\vec{v}).$
	Therefore, we can bound the welfare loss as required by $v_i\cdot (x_i^\textrm{A}(\vec{v})-x_i^\textrm{A}(\vec{s})) \leq p_\textrm{A}(\vec{v}) \cdot x_i^\textrm{A}(\vec{v}) \leq B_i,$ where the last inequality follows from the definition of the budget.
\end{proof}

We also measure the revenue loss due to the dynamic. The heart of the proof is in the following argument: if we remove the last deviator from the market, we obtain a \emph{reduced market} with $n-1$ buyers, for which we can show the equilibrium price of the original market is an envy-free price. 

We will say that a mechanism $\textrm{A}$ is a $\beta$ approximation for the revenue objective if for every market $\mathcal{M} = (\vec{v},\vec{B},m)$, the revenue achieved by $\textrm{A}$ on $\mathcal{M}$ is at least $\beta$ times the maximum revenue achieved by \emph{any envy-free pricing} pair $(\vec{x},p)$ on $\mathcal{M}$.

\begin{theorem}[Revenue of the dynamic from truth-telling]\label{thm:revenue-guarantee}
	Let $\textrm{A}$ be a mechanism that approximates the optimal revenue within a factor of $\beta \in [0,1]$  for every market. Consider any market $\mathcal{M} = (\vec{v}, \vec{B}, m)$. Then in every Nash equilibrium of $\textrm{A}$ reached through a best-response dynamic from truth-telling, the revenue is a $\left(\beta-\alpha\right)/2$ approximation of the optimal revenue for $\mathcal{M}$, where $\alpha$ is the budget share of the market.
\end{theorem}
\begin{proof}
	Given the market $\mathcal{M} = (\vec{v}, \vec{B}, m)$, for any strategy profile $\vec{\tilde{v}}$ we define the following notation. Let $p_\textrm{A}(\vec{\tilde{v}})$ denote the price and $x_i^\textrm{A}(\vec{\tilde{v}})$ the allocation to buyer $i$ computed by mechanism $A$ on input $\mathcal{\tilde{M}}= (\vec{\tilde{v},B},m)$. Additionally, let
	$\mathcal{REV}_{A}(\vec{\tilde{v}},\vec{B})$ denote the revenue extracted by $A$ on $\mathcal{\tilde{M}}$.
	For any envy-free price $p$ of $\mathcal{\tilde{M}}$, let $\mathcal{REV}(\vec{\tilde{v}}, \vec{B},p)$ denote the maximum possible revenue at the price $p$. This is achieved by serving fully all the hungry buyers at $p$ and allocating as many units as possible to the semi-hungry buyers. Finally, let $\mathcal{REV}_0(\vec{\tilde{v}}, \vec{B},p)$ denote the least possible revenue at price $p$, which is attained by allocating zero units to each semi-hungry buyer, and $\mathcal{REV}(\vec{\tilde{v}}, \vec{B})$ denote the maximum possible revenue from the market $\mathcal{\tilde{M}}$. 
	For ease of notation, we will write $\mathcal{H}_{p_A}(\vec{v})$ to denote the set $\mathcal{H}_{p_A(\vec{v})}(\vec{v})$ of hungry buyers at the price $p_A(\vec{v})$ computed by $A$ on input $\vec{v}$:  (and similarly for the sets $\mathcal{S}_{p_A(\vec{v})}(\vec{v})$ and $\mathcal{I}_{p_A(\vec{v})}(\vec{v})$).
	
	\medskip
	
	Let $\textrm{A}$ be any mechanism that approximates the optimal revenue within a factor of $\beta$. We consider two cases for the price $p_\textrm{A}(\vec{s})$:
	
	\medskip
	
	\noindent	\emph{Case 1}. $p_\textrm{A}(\vec{s}) = p_\textrm{A}(\vec{v})$. A deviating buyer can find a strictly improving deviation if and only if it can strictly lower the price. Since the price reached is equal to the price output on the true valuations, it follows that no deviation has taken place.  Then $\vec{v}$ is the Nash equilibrium to which the dynamic converges and there is no loss in revenue.\\
	
	\noindent \emph{Case 2}. $p_\textrm{A}(\vec{s}) < p_\textrm{A}(\vec{v})$. The proof of Theorem \ref{thm:convergence} implies that 
	in the Nash equilibrium reached from the truth-telling profile, each buyer receives at least as many units as they did on the true input $\vec{v}$, except possibly the last deviator. Let $\ell$ be the last deviator in the best response path from $\vec{v}$ to $\vec{s}$.
	
	Consider the \emph{reduced market} $\mathcal{M}' = (\vec{v}_{-\ell}, \vec{B}_{-\ell},m)$, which is obtained from the market $\mathcal{M}$ by removing buyer $\ell$. Let $p_{min}^{-\ell}$ be the minimum envy-free price in $\mathcal{M}'$. Our goal is to lower bound the revenue attained by $\textrm{A}$ in the Nash equilibrium\footnote{Note that the revenue objective is not a function of the real values and therefore it can be measured by simply the reports and the corresponding prices.}, for which it will be sufficient to provide a lower bound on $\mathcal{REV}_0(\vec{s}, \vec{B})$ since $$\mathcal{REV}_{\textrm{A}}(\vec{s}, \vec{B},p_\textrm{A}(\vec{s})) \geq \mathcal{REV}_0(\vec{s}, \vec{B}, p_\textrm{A}(\vec{s})).$$
	
	\noindent First, we claim that $p_\textrm{A}(\vec{s}) \geq p_{min}^{-\ell}$, or equivalently, that $p_\textrm{A}(\vec{s})$ is an envy-free price for the market $\mathcal{M}'$. This follows from the following fact established by Theorem \ref{thm:convergence}:
	\begin{quote}
		\emph{During the best response process, the price always decreases and at any point in time, the only buyer that appears semi-hungry from the set $\mathcal{I}_{p_\textrm{A}}(\vec{v})$ (i.e. the set of interested buyers at $\vec{v}$), is the last deviator}.
	\end{quote}
	In our case, the last deviator is buyer $\ell$. Then for each buyer $i \in \mathcal{I}_{p_\textrm{A}}(\vec{v})  \setminus \{\ell\}$, we have that $s_i > p_\textrm{A}(\vec{s})$, while for buyer $\ell$ we have $s_{\ell} \geq p_\textrm{A}(\vec{s})$. 
	This implies that $$\mathcal{REV}_{0}(\vec{s}, \vec{B},p_\textrm{A}(\vec{s})) =\mathcal{REV}_{0}(\vec{v}_{-\ell},\vec{B}_{-\ell},p_\textrm{A}(\vec{s})).$$
	
	\noindent Thus, it suffices to bound the minimum revenue attainable at the price $p_\textrm{A}(\vec{s})$, that is, $$\mathcal{REV}_{0}(\vec{v}_{-\ell}, \vec{B}_{-\ell}, p_\textrm{A}(\vec{s})).$$
	
	\noindent Next, note that since $p_\textrm{A}(\vec{s}) \in [p_{min}^{-\ell},p_\textrm{A}(\vec{v}))$, we can obtain that $$\mathcal{REV}_{0}(\vec{v}_{-\ell},\vec{B}_{-\ell},p_\textrm{A}(\vec{s})) \geq (1/2)\mathcal{REV}(\vec{v}_{-\ell},\vec{B}_{-\ell},p_\textrm{A}(\vec{v})).$$ For ease of exposition, let $\alpha_i = B_i/p_\textrm{A}(\vec{s})$ and $\alpha_i^* = B_i/p_\textrm{A}(\vec{v})$, $\forall i \in N$. Denote by $L$ the set of buyers with valuations at least $p_\textrm{A}(\vec{v})$ in the market $\mathcal{M'}=(\vec{v}_{-\ell}, \vec{B}_{-\ell},m)$ (i.e. buyers $j$ with $v_j \geq p_\textrm{A}(\vec{v})$) that can afford at least one unit at $p_\textrm{A}(\vec{v})$; note that the set of buyers that get allocated any items at $p_\textrm{A}(\vec{s})$ is a superset of $L$. Additionally, since $p_\textrm{A}(\vec{v}) > p_\textrm{A}(\vec{s})$, the set $L$ does not contain any buyers that are semi-hungry at $p_\textrm{A}(\vec{s})$ on $\mathcal{M}'$. Moreover, the revenue at $p_\textrm{A}(\vec{v})$ is bounded by the revenue attained at the (possibly infeasible) allocation where all the buyers in $L$ get the maximum number of units in their demand.
	These observations give the next inequalities:
	$$\mathcal{REV}_{0}(\vec{v}_{-\ell},\vec{B}_{-\ell},p_\textrm{A}(\vec{s}))  \geq \sum_{i \in L} \left \lfloor \alpha_i \right \rfloor \cdot p_\textrm{A}(\vec{s})\ \  \text{ and } \ \ 
	\mathcal{REV}(\vec{v}_{-\ell},\vec{B}_{-\ell},p_\textrm{A}(\vec{v})) \leq \sum_{i \in L} \left \lfloor \alpha_i^* \right \rfloor \cdot p_\textrm{A}(\vec{v}).$$
	Then the revenue loss, denoted by $$r=\frac{\mathcal{REV}_{0}(\vec{v}_{-\ell},\vec{B}_{-\ell},p_\textrm{A}(\vec{s}))}{\mathcal{REV}(\vec{v}_{-\ell},\vec{B}_{-\ell},p_\textrm{A}(\vec{v}))},$$
	can be bounded as follows:
	\begin{small}
	\begin{eqnarray*}
		r&\geq&\frac{\sum_{i \in L} \left \lfloor \alpha_i \right \rfloor \cdot p_\textrm{A}(\vec{s})}{\sum_{i \in L} \left \lfloor \alpha_i^* \right \rfloor \cdot p_\textrm{A}(\vec{v})}
		\geq  
		\frac{\sum_{i \in L} \left \lfloor \alpha_i  \right \rfloor  \cdot p_\textrm{A}(\vec{s})}{\sum_{i \in L} \alpha_i^* \cdot p_\textrm{A}(\vec{v})}
		= \frac{\sum_{i \in L} \left \lfloor \alpha_i \right \rfloor  \cdot p_\textrm{A}(\vec{s})}{\sum_{i \in L} B_i}
		= 
		\frac{\sum_{i \in L} \left \lfloor \alpha_i \right \rfloor }{\sum_{i \in L} \alpha_i } 
		\geq 
		\frac{\sum_{i \in L} \left \lfloor \alpha_i \right \rfloor}{\sum_{i \in L} 2 \left \lfloor \alpha_i \right \rfloor} 
		= \frac{1}{2},
	\end{eqnarray*}
\end{small}
	where we used the fact that the for any buyer $i \in L$, $\left \lfloor \alpha_i \right \rfloor \geq 1$, and so 
	$\alpha_i  \leq \left \lfloor \alpha_i \right \rfloor + 1 \leq 2 \left \lfloor \alpha_i \right \rfloor$, since by construction, all buyers in $L$ can afford at least one unit
	at $p_\textrm{A}(\vec{v})$ and therefore at $p_\textrm{A}(\vec{s})$. Note that in case the set $L$ is empty, then Mechanism $\textrm{A}$ extracts zero revenue at the truth-telling profile and the theorem follows trivially.
	Finally, observe that $\mathcal{REV}(\vec{v}_{-\ell},\vec{B}_{-\ell},p_\textrm{A}(\vec{v})) \geq \mathcal{REV}(\vec{v},\vec{B},p_\textrm{A}(\vec{v}))-B_{\ell}.$ This inequality holds because $p_\textrm{A}(\vec{v})$ is an envy-free price on input $\mathcal{M'}=(\vec{v}_{-\ell}, \vec{B}_{-\ell},m)$ and outputting this price results in a loss of revenue of at most the budget of the removed buyer. Since Mechanism $\textrm{A}$ outputs price $p_\textrm{A}(\vec{v})$ on $(\vec{v},\vec{B})$ and since it is a $\beta$-approximation mechanism for the revenue objective, it also holds that $\mathcal{REV}(\vec{v}, \vec{B},p_\textrm{A}(\vec{v})) \geq \mathcal{REV}_{\textrm{A}}(\vec{v}, \vec{B})\geq \beta \cdot \mathcal{REV}(\vec{v}, \vec{B}).$
	
	Tying everything together we have:
	\begin{eqnarray*}
		\mathcal{REV}_{\textrm{A}}(\vec{s}, \vec{B}) &\geq& \mathcal{REV}_{0}(\vec{s}, \vec{B},p_\textrm{A}(\vec{s})) =\mathcal{REV}_{0}(\vec{v}_{-\ell},\vec{B}_{-\ell},p_\textrm{A}(\vec{s})) 
		 \geq  \frac{1}{2} \cdot \mathcal{REV}(\vec{v}_{-\ell},\vec{B}_{-\ell},p_\textrm{A}(\vec{v})) \\ &\geq& \frac{\beta}{2}\cdot \mathcal{REV}(\vec{v}, \vec{B}) - \frac{B_{\ell}}{2} 
		 \geq  \frac{1}{2}\left(\beta-\alpha\right)\mathcal{REV}(\vec{v}, \vec{B}), 
	\end{eqnarray*}
	where the last inequality follows from the budget share definition.
	This completes the proof.
\end{proof}
Note that as $\beta$ approaches $1$ (that is, $\textrm{A}$ approaches a revenue-optimal mechanism) and $\alpha$ approaches $0$ (that is the market becomes fully competitive), the revenue attained in the equilibrium is at least half of the optimum. 

As can be seen from the next construction, the budget share is necessary. If there is one buyer with a budget share of $100\%$, then the revenue obtained as a result of this dynamic can be arbitrarily worse compared to the truth-telling profile (even for a revenue maximizing mechanism).

\begin{theorem}\label{badrevenuemonop}
	There is a mechanism (even a revenue maximizing one) for Walrasian envy-free pricing such that for all $\epsilon > 0$ \footnote{For the lower bound we fix the step size to 1 on the input and output grids.}, the best response dynamic starting from the truthful profile converges to a Nash equilibrium where the revenue is $\Omega(1/\epsilon)$ times worse than the optimum on some market.
\end{theorem}

\section{Equilibrium Set}\label{sec:nashset}

In this section, we consider all the equilibria of mechanisms, not only those reachable by best-response dynamics starting from the truth-telling profile. These states are fixed points of the dynamic. We will see that for a wide class of natural mechanisms, most of the guarantees carry over to all ``reasonable'' (i.e. non-overbidding) equilibria. 
``Overbidding" refers to the behavior in which a buyer $i$, with true value $v_i$, declares a bid (value) higher than its true value per unit. Any overbidding strategy is weakly dominated by truth-telling and can be ruled out using arguments about uncertainty, risk-aversion or trembling-hand considerations\cite{lucier2010price,caragiannis2015bounding}. For this reason, overbidding has been ruled out as ``unnatural'' in the literature (e.g. see \cite{feldman2013simultaneous}) and the study of no-overbidding equilibria is common (see e.g. the incentive properties of the second-price auction) \cite{christodoulou2008bayesian,lucier2010price,caragiannis2015bounding,bhawalkar2011welfare}.\footnote{For a more detailed discussion on the no-overbidding assumption and why it is natural, the reader is referred to \cite{lucier2010price}.} 
%
The omitted proofs of this section are in Appendix \ref{app:nashset}.
\subsection{Properties}

Given a market, a mechanism must determine two things: the price and the allocation to the semi-hungry buyers. The hungry buyers require a fixed number of units and the uninterested buyers no units. In the absence of incentives, the choice of how to allocate the semi-hungry buyers is not very important; any way of allocating the same number of units to them will give the same welfare and revenue. Interestingly however, we will show that the equilibrium behaviour of mechanisms is significantly affected by this choice and some choices result in better welfare and revenue guarantees. This is something that a mechanism designer could find useful when deciding the allocation rule for these buyers.

For the choice of prices, we will study mechanisms that follow a simple monotonicity property, that we will refer to as \emph{price-monotonicity}:
\begin{definition}[Price-monotone]
	An mechanism $\textrm{A}$ is \emph{price-monotone} if for any two valuation profiles $\mathbf{v}$ and $\vec{v}'$ such that $v'_i \leq v_i$ for all buyers $i \in N$ (with the other market parameters fixed), it holds that $p_{\textrm{A}}(\vec{v}') \leq p_{\textrm{A}}(\mathbf{v})$.
\end{definition}

We have the following definitions regarding the manner in which a mechanism allocates units to the semi-hungry buyers. 

\begin{definition}[Non-wasteful]
	An mechanism $\textrm{A}$ is \emph{non-wasteful} if it always allocates as many units as possible to the semi-hungry buyers.\footnote{Equivalently, the mechanism gives each semi-hungry buyer the maximum element in its demand set or sells all the units.}
\end{definition}

\begin{definition}[$\mathcal{S}$-Greedy]
	An mechanism $\textrm{A}$ is \emph{$\mathcal{S}$-Greedy} if for any market, after allocating the hungry buyers, it fixes an ordering of the semi-hungry buyers and allocates as many units as possible to each buyer $i$ in the ordering until exhausting $i$'s budget or running out of units. 
\end{definition}

Note that the ordering selected by an $\mathcal{S}$-greedy mechanism can be different for different input profiles, even if the sets of semi-hungry buyers in the two inputs are the same.

\begin{definition}[Supply-monotone]
	An mechanism $A$ is said to be \emph{supply-monotone} if keeping the valuation of a buyer and the price unchanged, while increasing the supply, results in that buyer receiving (weakly) more units. 
\end{definition}
\noindent In other words, for any budgets $\vec{B}$ and valuation profiles $\vec{v} = (v_i, \vec{v}_{-i})$ and $\vec{v}' = (v_i, \vec{v}_{-i}')$, if $A$ computes the same price $p$ on markets $\mathcal{M}= (\vec{v}, \vec{B},m)$ and 
$\mathcal{M}= (\vec{v}', \vec{B},m)$, and there are more units available for buyer $i$ in $v'$, then $x_i' \geq x_i$, where $\vec{x}, \vec{x}'$ are the allocations made by $A$ on $\mathcal{M}$ and $\mathcal{M}'$ respectively.

\noindent Supply-monotonicity is a property satisfied by many natural mechanisms. We have the implications: 
\[ \mathcal{S\textit{-greedy}} \Rightarrow \textit{Non-wasteful} \Rightarrow  \textit{Supply monotone}.\]
Therefore, proving a result for the class of all supply-monotone mechanisms immediately implies a result for all non-wasteful mechanisms (and hence all $\mathcal{S\text{-greedy}}$ mechanisms).
We will use the term \emph{monotone} for a mechanism that satisfies the monotonicity properties in the prices and the allocation to the semi-hungry buyers.

\begin{definition}[Monotone]
	A (Walrasian) envy-free pricing mechanism is \emph{monotone} if it's both price-monotone and supply-monotone.
\end{definition}

\subsection{Social Welfare and Revenue}

In this subsection, we state our main theorems that bound the welfare loss and the revenue approximation for any non-overbidding equilibrium of monotone mechanisms, which will then give us corollaries for the classes of mechanisms discussed in the previous subsection. We will actually prove quantified versions of the welfare and revenue guarantees, which are further parametrized by the number of semi-hungry buyers that receive partial allocations. 

More concretely, we will let $\mathcal{U}_{p_{\textrm{A}}}(\mathbf{v}) \subseteq \mathcal{S}_{p_{\textrm{A}}}(\mathbf{v})$ denote the set of semi-hungry buyers that receive partial allocations by Mechanism $\textrm{A}$ at price $p_{\textrm{A}}$, i.e. for each $i \in \mathcal{U}_{p_{\textrm{A}}}(\mathbf{v})$ it holds that $x_i^{\textrm{A}}(\vec{v}) \in (0,\min\{\floor{B_i/p},m\})$. Also, let $\gamma_A = \max_{\vec{v}}\mathcal{U}_{p_{\textrm{A}}}(\mathbf{v})$ be the maximum possible number of semi-hungry buyers with partial allocations over all possible inputs $\mathbf{v}$, where the known parameters ($n,m$ and $\vec{B}$) are fixed. Note that for an $\mathcal{S}$-Greedy mechanism,
it holds that $\gamma_A \leq 1$.

\subsubsection{Welfare guarantees}

First, we state the theorem that extends the welfare guarantees to the set of all non-overbidding equilibria.
\begin{theorem}[Welfare in any non-overbidding Nash equilibrium] \label{thm:welfare-extension}
	Let $\textrm{A}$ be a monotone mechanism. Then in any pure Nash equilibrium of $\textrm{A}$ where buyers do not overbid, the loss in social welfare (compared to the truth-telling outcome of $\textrm{A}$) is at most $\gamma_A \cdot B^*$, where $\gamma_A$ is the maximum number of semi-hungry buyers that receive partial allocations by $\mathrm{A}$ and $B^*$ the maximum budget.
\end{theorem}

\noindent Since supply-monotonicity implies several other properties, Theorem \ref{thm:welfare-extension} has the following corollary.

\begin{corollary}\label{thm:we-greedy}
	Let $\textrm{A}$ be a price-monotone, mechanism. Then in any pure Nash equilibrium of $\textrm{A}$ where buyers do not overbid, the loss in social welfare (compared to the truth-telling outcome of $\textrm{A}$) is at most 
	\begin{itemize}
		\item[-] $\gamma_A \cdot B^*$, \space \space \space \space \space \space \space \space if Mechanism $\textrm{A}$ is non-wasteful, where $B^*$ is the maximum budget,
		\item[-] $max_{i \in N}B_i$, \space \space  if Mechanism $\textrm{A}$ is $\mathcal{S}$-Greedy.
	\end{itemize}
	where $\gamma_A$ is the maximum number of buyers that receive partial allocations by $\mathrm{A}$.
\end{corollary}
\begin{proof}
	The proof follows from Theorem \ref{thm:welfare-extension}, the fact that supply-monotonicity is implied by these two properties and the observation that for any $\mathcal{S}$-Greedy mechanism $\textrm{A}$, it holds that $\gamma_A \leq 1$.
\end{proof}
Theorem \ref{thm:welfare-extension} and its corollary bound the welfare loss for a large class of mechanisms, compared to the outcome of the mechanism when buyers are truth-telling. In particular, the class of price-monotone mechanisms that are $\mathcal{S}$-greedy contains welfare-maximizing mechanisms. As we show in the following theorem, if we restrict our attention to welfare-maximizing mechanisms, we can obtain the same guarantee without any assumption on the allocation of units to the semi-hungry buyers. 

\begin{theorem}\label{thm:welfaremaximizing}
	Let $\textrm{A}$ be a price-monotone welfare-maximizing mechanism. Then in any pure Nash equilibrium of $\textrm{A}$ where buyers do not overbid, the loss in social welfare (compared to the truth-telling outcome) is at most the maximum budget.
\end{theorem}

\subsubsection{Revenue Guarantees}

Next, we will prove that the revenue guarantee of Theorem \ref{thm:revenue-guarantee} extends to all reasonable equilibria of the natural class of price-monotone mechanisms for different allocation rules for the semi-hungry buyers. Again, the guarantees will depend on the nature of this allocation rule.

\begin{theorem}[Revenue in any non-overbidding Nash equilibrium]\label{thm:revenue-extension}
	Let $\textrm{A}$ be a monotone mechanism that approximates the optimal revenue within a factor of $\beta$ (with $0\leq \beta \leq 1$). Then in every non-overbidding Nash equilibrium of $\textrm{A}$, the revenue is a $\left(\beta-\gamma_A \cdot \alpha\right)/2$ approximation of the optimal revenue for that instance, where $\alpha$ is the budget share of the market and $\gamma_A$ is the maximum number of buyers that receive partial allocations by $\mathrm{A}$.
\end{theorem}

\noindent Again, we have the following corollary.
\begin{corollary}
	Let $\textrm{A}$ be a price-monotone mechanism which approximates the optimal revenue within a factor of $\beta$ on every instance. Then in every non-overbidding Nash equilibrium of $\textrm{A}$, the revenue approximation achieved is
	\begin{itemize}
		\item [-] $\left(\beta-\gamma_A \cdot \alpha \right)/2$, \space\space\space\space\space if $A$ is non-wasteful,
		\item [-] $\left(\beta-\alpha \right)/2$, \space\space\space\space\space\space \space\space\space \space \space \space if $A$ is $\mathcal{S}$-Greedy,
	\end{itemize}   
	where $\alpha$ is the budget share of the market and $\gamma_A$ is the maximum number of semi-hungry buyers receiving partial allocations by $\mathrm{A}$.
\end{corollary}
\begin{proof}
	The proof follows directly from Theorem \ref{thm:revenue-extension}, together with the fact that supply-monotonicity is implied by those two properties and that for any $\mathcal{S}$-Greedy mechanism $A$, it holds that $\gamma_A \leq 1$. 
\end{proof}

\section{The \textsc{All-or-Nothing} mechanism}\label{sec:aon}

In this section we study a mechanism with good guarantees for revenue and welfare that was suggested in \cite{BFMZ16}, 
which was shown to be \emph{truthful}, i.e. for any agent $i$, truth-telling is always a best-response, for any strategy profile $s_{-i}$ of the remaining buyers.\footnote{In game-theoretic terms, this is the same as saying that truth-telling is a \emph{dominant strategy equilibrium}.} The omitted proofs of this section can be found in Appendix \ref{app:section5}.

\begin{definition}[\emph{\textsc{Mechanism All-or-Nothing:}}] Given as input the valuations of the buyers, let $p$ be the minimum envy-free price and $\vec{x}$ the allocation obtained as follows: 
	\begin{itemize}
		\item[-] For every hungry buyer $i$, set $x_i$ to its demand. 
		\item[-] For every buyer $i$ with $v_i<p$, set $x_i = 0$. 
		\item[-] For every semi-hungry buyer $i$, set $x_i = \lfloor B_i / p \rfloor$ if possible, otherwise set $x_i=0$ taking\\ the semi-hungry buyers in lexicographic order.
	\end{itemize}
\end{definition} 

\noindent The \textsc{All-or-Nothing} mechanism has several desirable properties.

\begin{theorem}[\cite{BFMZ16}]\label{thm:aonbranzei}
	The \textsc{All-Or-Nothing} mechanism is truthful and 
	\begin{itemize}
		\item has welfare is at least $w^* - B^{*}$, where $w^*$ is the optimal welfare and $B^*$ is the maximum budget.
		\item it approximates the optimal revenue within a (multiplicative) factor of at least $\min\left\{1/2, 1-\alpha\right\}$, where $\alpha$ is the budget share of the market.
	\end{itemize}
	 The mechanism is optimal for both the revenue and welfare objectives when the market is even mildly competitive\footnote{In \cite{BFMZ16}, the notion of market competition is actually slightly different from ours, but simple calculations can show that it is upper bounded by the budget share $\alpha$; we state the theorem in terms of $\alpha$ for consistency.} and its approximation for welfare converges to $1$ as the market becomes fully competitive.
\end{theorem}


\noindent However, nothing is known about the equilibrium set of the mechanism, or its convergence properties when instantiated from any arbitrary strategy profile. To this end, first we provide the following theorem, with a proof in Appendix \ref{app:section5}.  %

\begin{theorem} \label{thm:aon_converge}
	For $n=2$ buyers,
	the best response dynamic of the \textsc{All-or-Nothing} mechanism converges to a Nash equilibrium from any initial strategy profile.
\end{theorem}

\noindent We conjecture the mechanism converges for any number of buyers.

\begin{open} 
	Does the best response dynamic of the \textsc{All-or-Nothing} mechanism converge from any initial profile for any number of buyers?
\end{open}

Finally, taking advantage of our theorems for general classes of mechanisms developed in the previous sections, we provide the following stronger guarantee:

\begin{theorem}\label{thm:aonmain}
	For $n=2$ buyers, the best-response dynamic starting from \emph{any profile} of the \textsc{All-Or-Nothing} mechanism converges to a pure Nash equilibrium. 
	Moreover, in every non-overbidding equilibrium of \textsc{All-Or-Nothing}, for any number of buyers,
	\begin{itemize}
		\item [-]has welfare is at least $w^* - B^{*}$, where $w^*$ is the optimal welfare and $B^*$ is the maximum budget.
		\item [-]the revenue is approximated within a (multiplicative) factor of at least $\min \left\{0.25\right.$, $\left.(1-\alpha)/2\right\} - \alpha/2$ where  $\alpha$ is the budget share of the market.
	\end{itemize}
		In particular, in competitive markets (where $\alpha$ approaches $0$), mechanism \textsc{All-Or-Nothing} approximates the optimal revenue within a factor of $4$ in any non-overbidding equilibrium.
\end{theorem}
\noindent The proof of the welfare and revenue guarantees of Theorem \ref{thm:aonmain} follow from the remark above, together with Theorem \ref{thm:aonbranzei}, Theorem \ref{thm:welfare-extension} and Theorem \ref{thm:revenue-extension}, since the \textsc{All-Or-Nothing} mechanism is monotonic and the number of semi-hungry buyers with partial allocations is at most $1$, i.e. $\gamma \leq 1$. 
\begin{figure}[h!]
	\centering
\begin{subfigure}{0.47\textwidth}
		\includegraphics[width=1\textwidth, trim={0 0 0 0},clip]{AON_n25_m20_B120_truerev}
	\end{subfigure}
	\begin{subfigure}{0.47\textwidth}
		\includegraphics[width=1\textwidth, trim={0 0 0 0},clip]{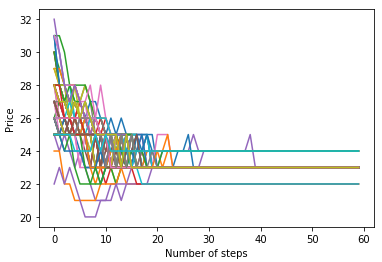}
	\end{subfigure}
		\caption{The convergence properties of the \textsc{All-Or-Nothing} mechanism when the dynamic starts from 100 different arbitrary profiles, for $n=25, m=20$. The budgets are drawn uniformly at random from $[1,50]$ (right) and 
		 $[1,125]$ (left). 
		The differently coloured lines indicate different starting profiles $\mathbf{s}^{(0)}$.} \label{figfig}
\end{figure}

\section{Discussion}

It would be interesting to understand more precisely the best response process initiated from arbitrary states (in particular, whether the \textsc{All-or-Nothing} mechanism converges for any number of buyers) and more general classes of valuations.



\addcontentsline{toc}{section}{\protect\numberline{}References}%

\bibliography{multiunitarxiv}

\newcommand{\etalchar}[1]{$^{#1}$}
\begin{thebibliography}{BFRMZ17}

\bibitem[AAE{\etalchar{+}}08]{awerbuch2008fast}
B.~Awerbuch, Y.~Azar, A.~Epstein, V.~S. Mirrokni, and A.~Skopalik.
\newblock Fast convergence to nearly optimal solutions in potential games.
\newblock In {\em EC}, pages 264--273. ACM, 2008.

\bibitem[AD54]{AD54}
K.~J. Arrow and G.~Debreu.
\newblock Existence of an equilibrium for a competitive economy.
\newblock {\em Econometrica}, 22(3):265--290, 1954.

\bibitem[Aus04]{Ausubel04}
L.~M. Ausubel.
\newblock An efficient ascending-bid auction for multiple objects.
\newblock {\em Am. Econ. Rev}, 94(5):1452--1475, 2004.

\bibitem[BBM08]{balcan2008item}
M.-F. Balcan, A.~Blum, and Y.~Mansour.
\newblock Item pricing for revenue maximization.
\newblock In {\em ACM EC}, pages 50--59, 2008.

\bibitem[BBN17]{BBN17}
M.~Babaioff, L.~Blumrosen, and N.~Nisan.
\newblock Selling complementary goods: Dynamics, efficiency and revenue.
\newblock In {\em ICALP}, 2017.

\bibitem[BCI{\etalchar{+}}05]{BCIMS05}
C.~Borgs, J.~Chayes, N.~Immorlica, M.~Mahdian, and A.~Saberi.
\newblock Multi-unit auctions with budget-constrained bidders.
\newblock In {\em ACM EC}, EC '05, pages 44--51, 2005.

\bibitem[BCMP13]{BCMP13}
S.~Br\^{a}nzei, I.~Caragiannis, J.~Morgenstern, and A.~D. Procaccia.
\newblock How bad is selfish voting?
\newblock In {\em AAAI}, pages 138--144, 2013.

\bibitem[BFRMZ17]{BFMZ16}
S.~Br\^{a}nzei, A.~Filos-Ratsikas, P.~B. Miltersen, and Y.~Zeng.
\newblock Walrasian pricing in multi-unit auctions.
\newblock In {\em MFCS}, pages 80:1--80:14, 2017.

\bibitem[BGN03]{bartal2003incentive}
Y.~Bartal, R.~Gonen, and N.~Nisan.
\newblock Incentive compatible multi unit combinatorial auctions.
\newblock In {\em TARK}, pages 72--87, 2003.

\bibitem[BK02]{binmore2002biggest}
K.~Binmore and P.~Klemperer.
\newblock The biggest auction ever: the sale of the british 3g telecom
  licences.
\newblock {\em Econ. J}, 112(478), 2002.

\bibitem[BLM09]{bulow2009winning}
J.~Bulow, J.~Levin, and P.~Milgrom.
\newblock Winning play in spectrum auctions, 2009.
\newblock Tech report.

\bibitem[BLNPL14]{babaioff2014efficiency}
M.~Babaioff, B.~Lucier, N.~Nisan, and R.~Paes~Leme.
\newblock On the efficiency of the walrasian mechanism.
\newblock In {\em ACM EC}, pages 783--800, 2014.

\bibitem[BP15]{BP15}
S.~Branzei and A.~Procaccia.
\newblock Verifiably truthful mechanisms.
\newblock In {\em ITCS}, pages 297--306, 2015.

\bibitem[BR11]{bhawalkar2011welfare}
K.~Bhawalkar and T.~Roughgarden.
\newblock Welfare guarantees for combinatorial auctions with item bidding.
\newblock In {\em SODA}, 2011.

\bibitem[BS00]{BS00}
W.~Brainard and H.~E. Scarf.
\newblock How to compute equilibrium prices in 1891.
\newblock {\em Cowles Foundation Discussion Paper}, 2000.

\bibitem[CD11]{CD11}
N.~Chen and X.~Deng.
\newblock On nash dynamics of matching market equilibria.
\newblock 03 2011.

\bibitem[CDE{\etalchar{+}}08]{cary2008best}
M.~Cary, A.~Das, B.~Edelman, I.~Giotis, K.~Heimerl, A.~R. Karlin, C.~Mathieu,
  and M.~Schwarz.
\newblock On best-response bidding in gsp auctions, 2008.
\newblock Tech report.

\bibitem[CDE{\etalchar{+}}14]{CDEG+14}
M.~Cary, A.~Das, B.~Edelman, I.~Giotis, K.~Heimerl, A.~R. Karlin, S.~D.
  Kominers, C.~Mathieu, and M.~Schwarz.
\newblock Convergence of position auctions under myopic best-response dynamics.
\newblock {\em ACM Trans. Econ. Comput.}, 2(3):9:1--9:20, July 2014.

\bibitem[CKK{\etalchar{+}}15]{caragiannis2015bounding}
I.~Caragiannis, C.~Kaklamanis, P.~Kanellopoulos, M.~Kyropoulou, B.~Lucier,
  R.~P. Leme, and E.~Tardos.
\newblock Bounding the inefficiency of outcomes in generalized second price
  auctions.
\newblock {\em JET}, 156:343--388, 2015.

\bibitem[CKS08]{christodoulou2008bayesian}
G.~Christodoulou, A.~Kov{\'a}cs, and M.~Schapira.
\newblock Bayesian combinatorial auctions.
\newblock In {\em ICALP}, pages 820--832, 2008.

\bibitem[CS07]{chien2007convergence}
S.~Chien and A.~Sinclair.
\newblock Convergence to approximate nash equilibria in congestion games.
\newblock In {\em SODA}, pages 169--178, 2007.

\bibitem[CS08]{cheung2008approximation}
M.~Cheung and C.~Swamy.
\newblock Approximation algorithms for single-minded envy-free
  profit-maximization problems with limited supply.
\newblock In {\em FOCS}, pages 35--44, 2008.

\bibitem[CT16]{cole2015does}
R.~Cole and Y.~Tao.
\newblock When does the price of anarchy tend to 1 in large walrasian auctions
  and fisher markets?
\newblock In {\em EC}, 2016.

\bibitem[DHP17]{DHP17}
N.~R. Devanur, N.~Haghpanah, and C.-A. Psomas.
\newblock Optimal multi-unit mechanisms with private demands.
\newblock In {\em EC}, pages 41--42, 2017.

\bibitem[DK17]{dutting2017best}
P.~D{\"u}tting and T.~Kesselheim.
\newblock Best-response dynamics in combinatorial auctions with item bidding.
\newblock In {\em SODA}, pages 521--533, 2017.

\bibitem[dKMST13]{de2013inefficiency}
Bart de~Keijzer, Evangelos Markakis, Guido Schafer, and Orestis Telelis.
\newblock Inefficiency of standard multi-unit auctions.
\newblock In {\em ESA}, pages 385--396. Springer, 2013.

\bibitem[DLN12]{dobzinski2012multi}
S.~Dobzinski, R.~Lavi, and N.~Nisan.
\newblock Multi-unit auctions with budget limits.
\newblock {\em GEB}, 74(2):486--503, 2012.

\bibitem[DN07]{dobzinski2007mechanisms}
S.~Dobzinski and N.~Nisan.
\newblock Mechanisms for multi-unit auctions.
\newblock In {\em ACM EC}, pages 346--351, 2007.

\bibitem[DN15]{dobzinski2015multi}
S.~Dobzinski and N.~Nisan.
\newblock Multi-unit auctions: beyond roberts.
\newblock {\em Journal of Economic Theory}, 156:14--44, 2015.

\bibitem[DS16]{DS16}
C.~Daskalakis and V.~Syrgkanis.
\newblock Learning in auctions: Regret is hard, envy is easy.
\newblock In {\em FOCS}, pages 219--228, 2016.

\bibitem[FFGL13]{feldman2013simultaneous}
M.~Feldman, H.~Fu, N.~Gravin, and B.~Lucier.
\newblock Simultaneous auctions are (almost) efficient.
\newblock In {\em STOC}, pages 201--210, 2013.

\bibitem[FFLS12]{feldman2012revenue}
M.~Feldman, A.~Fiat, S.~Leonardi, and P.~Sankowski.
\newblock Revenue maximizing envy-free multi-unit auctions with budgets.
\newblock In {\em ACM EC}, pages 532--549, 2012.

\bibitem[FLSS11]{fiat2011single}
A.~Fiat, S.~Leonardi, J.~Saia, and P.~Sankowski.
\newblock Single valued combinatorial auctions with budgets.
\newblock In {\em EC}, pages 223--232, 2011.

\bibitem[GHK{\etalchar{+}}05]{guruswami2005profit}
V.~Guruswami, J.~Hartline, A.~Karlin, D.~Kempe, C.~Kenyon, and F.~McSherry.
\newblock On profit-maximizing envy-free pricing.
\newblock In {\em SODA}, pages 1164--1173, 2005.

\bibitem[GML13]{GMP13}
G.~Goel, V.~S. Mirrokni, and R.~P. Leme.
\newblock Clinching auctions with online supply.
\newblock In {\em SODA}, pages 605--619, 2013.

\bibitem[GS99]{GS99}
F.~Gul and E.~Stacchetti.
\newblock Walrasian equilibrium with gross substitutes.
\newblock {\em Journal of Economic Theory}, 87(1):95--124, 1999.

\bibitem[HKMN11]{HKMN11}
A.~Hassidim, H.~Kaplan, Y.~Mansour, and N.~Nisan.
\newblock Non-price equilibria in markets of discrete goods.
\newblock In {\em EC}, pages 295--296, 2011.

\bibitem[HN13]{HN12}
S.~Hart and N.~Nisan.
\newblock The menu-size complexity of auctions.
\newblock In {\em ACM EC}, pages 565--566, 2013.

\bibitem[HR09]{hartline2009simple}
J.~D. Hartline and T.~Roughgarden.
\newblock Simple versus optimal mechanisms.
\newblock In {\em ACM EC}, pages 225--234, 2009.

\bibitem[KC82]{KC82}
A.~S. Kelso and V.~P. Crawford.
\newblock Job matching, coalition formation, and gross substitutes.
\newblock {\em Econometrica}, 50, 1982.

\bibitem[Kes16]{Kesselheim16}
T.~Kesselheim.
\newblock Combinatorial auctions with item bidding: Equilibria and dynamics.
\newblock In {\em WINE}, 2016.

\bibitem[Kle10]{klemperer2010product}
P.~Klemperer.
\newblock The product-mix auction: A new auction design for differentiated
  goods.
\newblock {\em J. Eur. Econ. Assoc}, 8(2-3), 2010.

\bibitem[Laf87]{Laffont87}
J.~Laffont.
\newblock Incentives and the allocation of public goods.
\newblock {\em Handbook of Public Economics}, 2:537–569, 1987.

\bibitem[LB10]{lucier2010price}
B.~Lucier and A.~Borodin.
\newblock Price of anarchy for greedy auctions.
\newblock In {\em SODA}, pages 537--553, 2010.

\bibitem[Li17]{Li17}
S.~Li.
\newblock Obviously strategy-proof mechanisms.
\newblock {\em Am. Econ. Rev}, 107(11):3257--87, 2017.

\bibitem[LR96]{laffont1996optimal}
J.~Laffont and J.~Robert.
\newblock Optimal auction with financially constrained buyers.
\newblock {\em Economics Letters}, 52(2):181--186, 1996.

\bibitem[MPRJ10]{MPRJ10}
R.~Meir, M.~Polukarov, J.~Rosenschein, and N.~Jennings.
\newblock Convergence to equilibria in plurality voting.
\newblock In {\em AAAI}, pages 823--828, 2010.

\bibitem[MSZ15]{monaco2015revenue}
G.~Monaco, P.~Sankowski, and Q.~Zhang.
\newblock Revenue maximization envy-free pricing for homogeneous resources.
\newblock In {\em IJCAI}, pages 90--96, 2015.

\bibitem[MT12]{MT12}
E.~Markakis and Orestis Telelis.
\newblock {\em Uniform Price Auctions: Equilibria and Efficiency}, pages
  227--238.
\newblock 2012.

\bibitem[Nis14]{Nisan14}
N.~Nisan.
\newblock Algorithmic mechanism design through the lens of multi-unit auctions.
\newblock {\em Handbook of Game Theory 4}, 2014.

\bibitem[NRTV07]{AGT_book}
N.~Nisan, T.~Roughgarden, E.~Tardos, and V.~Vazirani.
\newblock {\em {Algorithmic Game Theory}}.
\newblock Cambridge Univ. Press, (editors) 2007.

\bibitem[NSVZ11]{nisan2011best}
N.~Nisan, M.~Schapira, G.~Valiant, and A.~Zohar.
\newblock Best-response auctions.
\newblock In {\em EC}, pages 351--360. ACM, 2011.

\bibitem[Pet04]{Petersen}
{M.A.} Petersen.
\newblock Information: Hard and soft.
\newblock Working paper, 7 2004.

\bibitem[Rou09]{roughgarden2009intrinsic}
T.~Roughgarden.
\newblock Intrinsic robustness of the price of anarchy.
\newblock In {\em STOC}, pages 513--522. ACM, 2009.

\bibitem[RST17]{RST17}
T.~Roughgarden, V.~Syrgkanis, and E.~Tardos.
\newblock The price of anarchy in auctions.
\newblock {\em J. Artif. Intell. Res.}, 59:59--101, 2017.

\bibitem[Wal74]{Walras74}
L.~Walras.
\newblock Elements d'economie politique pure, ou theorie de la richesse
  sociale.
\newblock 1874.
\newblock English trans.: Elements of pure economics; or, the theory of social
  wealth. American Economic Association and the Royal Economic Society, 1954.

\end{thebibliography}

\appendix

\section{Dynamics} \label{app:nashdynamic}

\begin{theorem}[Necessary allocation loss]\label{thm:appnecessaryloss}[\ref{thm:necessaryloss} in main text]
	There exist markets where the best response dynamic of a mechanism converges to an allocation where some buyer loses units compared to the truth-telling outcome.
\end{theorem}
\begin{proof}
	Let $\mathrm{A}$ be the mechanism which sets the price equal to the lowest valuation that is envy-free, breaking ties in a fixed way (the same for all valuations) to decide the order in which the semi-hungry buyers are allocated. Note that there is always such an envy-free price, since it is always possible to set it to the highest valuation, at which all the buyers are either semi-hungry or not interested. 
	Given a market $\mathcal{M}= (\vec{v}, \vec{B}, m)$, let $p_\textrm{A}(\vec{v})$ denote the price and $x_i^\textrm{A}(\vec{v})$ the allocation of buyer $i$ computed by $A$ for the market $\mathcal{M}$.
	
	Take a market with $n=3$ buyers---Alice, Bob, and Carol---and $m=3$ items, where the valuations are $v_{Alice} = 1.1$, $v_{Bob} = 1.1$, and $v_{Carol} = 1$, while the budgets are $B_{Alice} = B_{Bob} = 2.2$ and $B_{Carol} = 1$. Let the tie-breaking order be $1,3,2$ and the grid step $0.001$. At the truth-telling outcome, the price set by Mechanism $\mathrm{A}$ is $p_\mathrm{A}(\vec{v}) = v_{Bob} = 1.1$, where Alice gets $x_{Alice}^\mathrm{A}(\vec{v}) = 2$ units, Bob gets $x_{Bob}^\mathrm{A}(\vec{v}) = 1$ units, and Carol gets $x_{Carol}^\mathrm{A}(\vec{v}) = 0$ units. This is the lowest envy-free price equal to some valuation since setting the price to $v_{Carol} = 1$ would result in both Alice and Bob having a demand of two, while only three units are available. Alice's utility is
	$$u_{Alice}(\vec{v},p_\mathrm{A}(\vec{v})) = x_{Alice}^\mathrm{A}(\vec{v}) \cdot v_{Alice} - p_\mathrm{A}(\vec{v}) \cdot x_{Alice}^\mathrm{A}(\vec{v}) = 2 \cdot 1.1 - 1.1 \cdot 2 = 0.$$
	
	We claim that Alice has a best response at $v_{Alice}' = 1$. Let $\vec{v}' = (v_{Alice}', v_{Bob}, v_{Carol})$. 
	The price output by $\mathrm{A}$ at $\vec{v}'$ is $p_\mathrm{A}(\vec{v'}) = 1$. Since Bob appears hungry, while Alice and Carol are semi-hungry at $p_\mathrm{A}(\vec{v'})$ with tie-breaking order $1,3,2$, the allocation is $x_{Alice}^\mathrm{A}(\vec{v'}) = 1$, $x_{Bob}^\mathrm{A}(\vec{v'}) = 2 $ and $x_{Carol}^\mathrm{A}(\vec{v'}) = 0$, which gives Alice a utility of 
	$$u_{Alice}(\vec{v}', p_\mathrm{A}(\vec{v'})) = x_{Alice}^\mathrm{A}(\vec{v'}) \cdot v_{Alice} - p_\mathrm{A}(\vec{v'}) \cdot x_{Alice}^\mathrm{A}(\vec{v'})= 1 \cdot 1.1 - 1 \cdot 1 = 0.1 > u_{Alice}(\vec{v},p_\mathrm{A}(\vec{v})) = 0.$$
	
	For any input $v_{Alice}'' \in (1, 1.1)$, the price would be set to $v_{Alice}''$, at which Alice would still get only one unit (since Bob is hungry in this range), but at a higher price compared to $v_{Alice}'$. At any input $v_{Alice}'' < 1$, the price would be set to $v_{Alice}''$ but Alice would get no units since the demands of the hungry buyers, Bob and Carol, will be satisfied first. 
	It follows that $v_{Alice}'$ is a best response for Alice given the state $\vec{v}$.
	
	We show that $\vec{v}'$ is a Nash equilibrium. Note that 
	$$u_{Bob}(\vec{v}', p_\mathrm{A}(\vec{v'})) = v_{Bob} \cdot x_{Bob}^\mathrm{A}(\vec{v'}) - p_\mathrm{A}(\vec{v'}) \cdot x_{Bob}^\mathrm{A}(\vec{v'}) = 1.1 \cdot 2 - 1 \cdot 2 = 0.2.$$
	
	For any alternative input $v_{Bob}'' \in (1,1.1)$ of Bob, the price and allocation of Bob would remain fixed. The input $v_{Bob}'' = 1$ also does not change the price, and in fact gives Bob zero units, since now all the buyers appear semi-hungry and the order of serving them is Alice, Carol, Bob. Finally, any $v_{Bob}'' < 1$ results in price  equal to $v_{Bob}'' < 1$, at which Alice and Carol are hungry, thus giving Bob zero units; this cannot be an improvement. We argue that Carol also has no improving deviation. If Carol reports $v_{Carol}'' < 1$, the price would be decreased to $p'' = v_3''$; however, Alice and Bob appear hungry at this price, and would get all the units, leaving Carol with zero units. Thus $\vec{v}'$ is a Nash equilibrium in which Alice loses a unit compared to the truth-telling outcome, which completes the proof.
\end{proof}

\begin{theorem}[Convergence time of consistent mechanisms]\label{appdyn:conv}[\ref{dyn:conv} in main text]
	Let $\textrm{A}$ be any consistent mechanism. The best response dynamic starting from the truth-telling profile converges to a pure Nash equilibrium of $\textrm{A}$ in at most $n$ steps.
\end{theorem}
\begin{proof}
	Since convergence is established by Theorem \ref{thm:convergence}, we must show here that the convergence time is at most $n$ rounds. We will argue that each deviating buyer best-responds only once during the dynamic. Consider a market $\mathcal{M}= (\vec{v}, \vec{B}, m)$. For any valuation profile $\vec{\tilde{v}}$, let $p_\textrm{A}(\vec{\tilde{v}})$ denote the price and $x_i^\textrm{A}(\vec{\tilde{v}})$ the allocation of buyer $i$ computed by $A$ given as input the market $\mathcal{M}' = (\vec{\tilde{v},B},m)$ (i.e. the price output by the mechanism on valuations $\vec{\tilde{v}}$).
	Assume by contradiction that this is not true, and consider any buyer $i$ and two different best responses of the buyer: 
	
	\begin{itemize}
		\item the first deviation is from the truth-telling strategy $v_i$ to a different value $v'_{i}$, when the strategies of the other buyers are $\vec{s_{-i}}$.
		\item the second deviation is from $v'_{i}$ to a value $v''_{i}$, when the other buyers' strategies are $\vec{s'_{-i}}$.
	\end{itemize}
	By the proof of Theorem \ref{thm:convergence}, we know that the price decreases in each step of the best-response process, therefore we can assume that $$p_\textrm{A}(v_i,\vec{s_{-i}}) > p_{\textrm{A}}(v'_i,\vec{s_{-i}}) > p_{\textrm{A}}(v'_i,\vec{s'_{-i}})> p_{\textrm{A}}(v''_i,\vec{s'_{-i}}).$$ Consider the number of units allocated to buyer $i$ at those prices. First, it holds that $x_i^\textrm{A}(v'_i,\vec{s'_{-i}}) \geq x_i^\textrm{A}(v'_i,\vec{s_{-i}})$, since buyer $i$ appears to be hungry at price $p_{\textrm{A}}(v'_i,\vec{s'_{-i}})$ and $p_{\textrm{A}}(v'_i,\vec{s'_{-i}}) < p_{\textrm{A}}(v'_i,\vec{s_{-i}})$, so the buyer can not receive fewer units. Since the buyer is best-responding from $v'_i$ to $v''_i$, we have
	\begin{eqnarray*}
		(v_i   - p_\textrm{A}(v''_i,\vec{s'_{-i}})) \cdot x_i^\textrm{A}(v''_i,\vec{s'_{-i}}) > (v_i   - p_\textrm{A}(v'_i,\vec{s'_{-i}})) \cdot x_i^\textrm{A}(v'_i,\vec{s'_{-i}})
		> (v_i   - p_\textrm{A}(v'_i,\vec{s_{-i}})) \cdot x_i^\textrm{A}(v'_i,\vec{s_{-i}}),
	\end{eqnarray*}
	where the last inequality holds since $x_i^\textrm{A}(v'_i,\vec{s'_{-i}}) \geq x_i^\textrm{A}(v'_i,\vec{s_{-i}})$ and $p_\textrm{A}(v_i',\vec{s_{-i}}) > p_\textrm{A}(v'_i,\vec{s'_{-i}})$.
	
	Now consider the input profile $(v''_i,\vec{s_{-i}})$ and notice that the sets of hungry and semi-hungry buyers at price $p_\textrm{A}(v_i'',\vec{s'_{-i}})$ are the same as those on profile $(v_i'',\vec{s'_{-i}})$; this is implied by the properties of the best-response sequence of Theorem \ref{thm:convergence}. Therefore, by consistency, it is the case that $p_\textrm{A}(v''_i,\vec{s_{-i}})=p_\textrm{A}(v''_i,\vec{s'_{-i}})$ and $x_i^\textrm{A}(v''_i,\vec{s_{-i}}) = x_i^\textrm{A}(v''_i,\vec{s'_{-i}})$, which by the formula above means that 
	$$(v_i   - p_\textrm{A}(v''_i,\vec{s'_{-i}})) \cdot x_i^\textrm{A}(v''_i,\vec{s'_{-i}}) > (v_i   - p_\textrm{A}(v'_i,\vec{s_{-i}})) \cdot x_i^\textrm{A}(v'_i,\vec{s_{-i}}).$$
	
	This contradicts the fact that $v_i'$ was a best response on the input profile $(v_i,\vec{s_{-i}})$ in the first place. Therefore, each buyer best-responds at most once during the dynamic and the convergence is guaranteed within $n$ steps.
\end{proof}

\begin{theorem}\label{appbadrevenuemonop}[\ref{badrevenuemonop} in main text]
	There is a mechanism (even a revenue maximizing one) for Walrasian envy-free pricing such that for all $\epsilon > 0$ \footnote{For the lower bound we fix the step size to 1 on the input and output grids.}, the best response dynamic starting from the truthful profile converges to a Nash equilibrium where the revenue is $\Omega(1/\epsilon)$ times worse than the optimum on some market.
\end{theorem}
\begin{proof}
	Let $\mathrm{A}$ be a revenue optimal mechanism, such that for each input, if given a choice, it minimizes the number of allocated buyers (given that the revenue has been maximized) and sets the highest possible envy-free price with these properties.
	Moreover, let $\mathrm{A}$ break ties lexicographically when allocating to semi-hungry buyers, such that each buyer is allocated as many units it can afford before allocating the next buyer.
	Let both the input and output grids have step size $1$.
	
	Given $\epsilon > 0$, consider a market with $n=2$ buyers, $m = 1$ units, valuations $v_1 = \lceil 1/\epsilon \rceil$, $v_2 = 1$, and budgets $B_1 = B_2 = \lceil 1/\epsilon \rceil$.  On this input, $\mathrm{A}$ sets the price to $p = \lceil 1/\epsilon \rceil$, allocating $x_1 = 1$ units to buyer $1$ and $x_2 = 0$ units to buyer $2$. The revenue obtained this way is $\mathcal{REV}(v_1,v_2) = \lceil 1/\epsilon \rceil$, thus exhausting the budget of buyer $1$. This is the optimal revenue since there is only one unit and $\lceil 1/\epsilon \rceil$ is the maximum price any buyer is willing to pay for it.
	
	Buyer $1$ can respond from the truthful state with $v_1' = 1$. Then $\mathrm{A}$ sets the price to $p'=1$, allocating $x_1' = 1$ units to buyer $1$ and $x_2' = 0$ units to buyer $2$. This clearly improves buyer $1$'s utility, since he gets the same number of units at a lower price. The revenue at this price is $\mathcal{REV}(v_1',v_2) = 1$. We argue that no valuation $v_1''< 1$ can improve buyer $1$'s utility, since at any such profile $(v_1'', v_2)$ the mechanism can still set $p' = 1$ and sell the unit to buyer $2$ instead, which would give buyer $1$ a utility of zero. Moreover, at any $v_1' \in (1, v_1)$, buyer $1$ would still get the item for a price weakly higher than $1$, which would result in utility at most that obtained on input $(v_1'', v_2)$. Thus the report $v_1''$ is a best response for buyer $1$.
	

	Moreover, the state $(v_1', v_2)$ is a Nash equilibrium. The valuation $v_1'$ is a best response for buyer $1$, while if buyer $2$ reported a lower value than $1$, the price would remain the same without improving $2$'s allocation. We have 
	$$
	\frac{\mathcal{REV}(v_1,v_2)}{\mathcal{REV}(v_1',v_2)} = \left\lceil \frac{1}{\epsilon} 
	\right\rceil
	$$
	This completes the proof.
\end{proof}

\section{Equilibrium Set} \label{app:nashset}

We start with the following lemma, which will be used throughout the proofs. The lemma essentially states that for monotone mechanisms, if a buyer received any units under truth-telling, it must also receive some units at the equilibrium $\vec{s}$.

\begin{lemma}\label{lem:equil-one}
	Let $\textrm{A}$ be a monotone mechanism and $\vec{v}$ the true valuations of some market. Then in any non-overbidding pure Nash equilibrium $\vec{s}$ of $\textrm{A}$, there does not exist a buyer $i$ such that: 
	\begin{itemize}
		\item[-] 
		the buyer receives non-zero units under the truth-telling profile, 
		\item[-] 
		the buyer receives zero units in the equilibrium allocation. 
	\end{itemize} 
\end{lemma}
\begin{proof}
	Let $p_\mathrm{A}(\vec{\tilde{v}})$ denote the price chosen by Mechanism $\mathrm{A}$ and $x_i^{\textrm{A}}(\vec{\tilde{v}})$ denote the number of units allocated to buyer $i$ on input market $(\vec{\tilde{v}},\vec{B},m)$. 
	
	Assume by contradiction that there exists a buyer $i$ such that $x_i^{\textrm{A}}(\mathbf{v})>0$ and $x_i^{\textrm{A}}(\vec{s})=0$
	and consider the strategy $v_i$ of the buyer where it deviates to truth-telling. By the price-monotonicity of $\textrm{A}$, it holds that $p_{\textrm{A}}(\mathbf{s}) \leq p_{\textrm{A}}{(v_i,\mathbf{s}_{-i})} \leq p_{\textrm{A}}(\vec{v})$, since $s_j \leq v_j$ for all $j \in N$, by the fact that profile $\mathbf{s}$ is a no-overbidding equilibrium. We consider two cases.\\
	
	\noindent \textbf{Case 1:} \emph{Either it holds that $p_{\textrm{A}}(v_i,\mathbf{s}_{-i}) < p_{\textrm{A}}(\mathbf{v})$ or buyer $i$ was hungry on the truth-telling profile $\mathbf{v}$.} In that case, the buyer receives a strictly positive allocation $x_i^\textrm{A}(v_i,\mathbf{s_{-i}})$ at price $p_{\mathbf{A}}(v_i,\mathbf{s}_{-i})$, obtaining positive utility and contradicting the fact that $\mathbf{s}$ is a pure Nash equilibrium.\\

	\noindent \textbf{Case 2:} \emph{It holds that $p_{\textrm{A}}(v_i,\mathbf{s}_{-i}) = p_{\textrm{A}}(\mathbf{v})$ and buyer $i$ is semi-hungry on the truth-telling profile $\mathbf{v}$}. In that case, the equilibrium condition is violated unless the buyer receives exactly $0$ units on $(v_i,\mathbf{s}_{-i})$ at price $p_{\textrm{A}}(v_i,\mathbf{s}_{-i})$. However, by the fact that buyers do not overbid, it holds that $s_j \leq v_j$ for all $j \in N$ and therefore, it also holds that $\mathcal{I}_{p_\textrm{A}}(v_i,\vec{s}_{-i}) \subseteq \mathcal{I}_{p_\textrm{A}}(\vec{v})$ and $\mathcal{H}_{p_\textrm{A}}(v_i,\vec{s}_{-i}) \subseteq \mathcal{H}_{p_\textrm{A}}(\vec{v})$. In other words, there is more available supply for the semi-hungry buyers on $(v_i,\vec{s_{-i}})$ compared to $\vec{v}$, and Mechanism $\mathrm{A}$ outputs the same price on both inputs. By supply monotonicity, since $x_i^\mathbf{A}(\mathbf{v}) >0$ holds, it also holds that $x_i^\textrm{A}(v_i,\mathbf{s}_{-i})>0$ and $v_i$ is a beneficial deviation on profile $\vec{s}$ contradicting the fact that $\vec{s}$ is a pure Nash equilibrium.
\end{proof}	

\noindent Using Lemma \ref{lem:equil-one}, we will prove our social welfare and revenue guarantees for the set of all non-overbidding equilibria of mechanisms which are price-monotone under different assumptions for the allocations of semi-hungry buyers. In short, if the mechanism is 
\emph{$\mathcal{S}$-Greedy}, 
then the welfare and revenue guarantees of Theorem \ref{thm:welfare-guarantee} and Theorem \ref{thm:revenue-guarantee} extend to the case of all non-overbidding equilibria. If $\textrm{A}$ is \emph{non-wasteful} (but not necessarily $\mathcal{S}$-Greedy), then we prove quantified versions of the welfare and revenue guarantees, which are further parametrized by the number of semi-hungry buyers that receive partial allocations.

Recall from the main text that $\mathcal{U}_{p_{\textrm{A}}}(\mathbf{v}) \subseteq \mathcal{S}_{p_{\textrm{A}}}(\mathbf{v})$ denotes the set of semi-hungry buyers that receive partial allocations by Mechanism $\textrm{A}$ at price $p_{\textrm{A}}$, i.e. for each $i \in \mathcal{U}_{p_{\textrm{A}}}(\mathbf{v})$ it holds that $x_i^{\textrm{A}}(\vec{v}) \in (0,\min\{\floor{B_i/p},m\})$. Also, $\gamma_A = \max_{\vec{v}}\mathcal{U}_{p_{\textrm{A}}}(\mathbf{v})$ is the maximum possible number of semi-hungry buyers with partial allocations over all possible inputs $\mathbf{v}$, where the known parameters ($n,m$ and $\vec{B}$) are fixed. Note that for an $\mathcal{S}$-Greedy mechanism,
it holds that $\gamma_A \leq 1$.

\begin{theorem}[Welfare in any non-overbidding Nash equilibrium] \label{thm:appwelfare-extension}[\ref{thm:welfare-extension} in main text]
	Let $\textrm{A}$ be a monotone mechanism. Then in any pure Nash equilibrium of $\textrm{A}$ where buyers do not overbid, the loss in social welfare (compared to the truth-telling outcome of $\textrm{A}$) is at most $\gamma_A \cdot B^*$, where $\gamma_A$ is the maximum number of semi-hungry buyers that receive partial allocations by $\mathrm{A}$ and $B^*$ the maximum budget.
\end{theorem}
\begin{proof}
	Similarly to before, for a fixed vector of budgets $\vec{B}$ and number of units $m$, let $p_\textrm{A}(\vec{\tilde{v}})$ denote the price and $x_i^\textrm{A}(\vec{\tilde{v}})$ denote the allocation buyer $i$  on input $\mathcal{M}= (\vec{\tilde{v},B},m)$.	
	
	Let $\mathbf{s}$ be a non-overbidding equilibrium of $\textrm{A}$, i.e. it holds that $s_j \leq v_j$ for all $j \in N$. Since $\textrm{A}$ is price-monotone, this implies that $p_{\textrm{A}}(\mathbf{s}) \leq p_{\textrm{A}}(\mathbf{v})$. Additionally, since $\textrm{A}$ is is also supply-monotone, by Lemma \ref{lem:equil-one}, there do not exist any buyers that receive positive allocations on the truth-telling profile $\mathbf{v}$ and zero allocations in the equilibrium $\mathbf{s}$. Therefore, the only loss in welfare is due to each buyer $i \in \mathcal{U}_{p_{\textrm{A}}}(\vec{s})$ that receives $x_i^{\textrm{A}}(\mathbf{s})$ units on $\mathbf{s}$ and $x_i^{\textrm{A}}(\mathbf{v})$ units on $\mathbf{v}$, with $0 < x_i^{\textrm{A}}(\mathbf{s}) < x_i^{\textrm{A}}(\mathbf{v})$. Since $p_{\textrm{A}}(\mathbf{s}) \leq p_{\textrm{A}}(\mathbf{v})$, these buyers are semi-hungry at $\mathbf{s}$, as otherwise they could not receive fewer units at a price which is not larger than before. 
	
	Consider any such buyer $i \in \mathcal{U}_{p_\textrm{A}}(\vec{s})$ and consider the deviation $v_i$ to truth-telling, and the resulting profile $(v_i,\mathbf{s_{-i}})$. We will consider two cases:
	\medskip
	
	\noindent \textbf{Case 1:} {Buyer $i$ is hungry on $\mathbf{v}$ under $p_{\textrm{A}}(\mathbf{v})$.} Since $s_j \leq v_j$ for all $j \in N$ and Mechanism $\textrm{A}$ is monotone, it also holds that $p_\textrm{A}(v_i,\mathbf{s_{-i}}) \leq p_\textrm{A}(\vec{v})$ and therefore buyer $i$ is hungry on $(v_i,\mathbf{s_{-i}})$ at price $p_{\textrm{A}}(v_i,\mathbf{s_{-i}})$ as well. This means buyer $i$ receives at least $x_i^\textrm{A}(\vec{v})$ units at $p_{\textrm{A}}(v_i,\mathbf{s_{-i}})$ on profile $(v_i,\mathbf{s_{-i}})$.
	
	\medskip
	
	\noindent\textbf{Case 2:} \emph{Buyer $i$ is semi-hungry on $\mathbf{v}$ under $p_\textrm{A}(\mathbf{v})$}. Since $s_j \leq v_j$ for all $j \in N$ and Mechanism $\textrm{A}$ is price-monotone, it also holds that $p_\textrm{A}(v_i,\mathbf{s_{-i}}) \leq p_\textrm{A}(\vec{v})$ and therefore buyer $i$ is either hungry or semi-hungry on $(v_i,\mathbf{s_{-i}})$ at price $p_\textrm{A}(v_i,\mathbf{s_{-i}})$. If it is hungry, again it must receive at least $x_i^\textrm{A}(\vec{v})$ units at $p_\textrm{A}(v_i,\mathbf{s_{-i}})$ on profile $(v_i,\mathbf{s_{-i}})$. If it is semi-hungry, this means that $p_\textrm{A}(v_i,\mathbf{s_{-i}})=p_\textrm{A}(\mathbf{v})$ and therefore its demand set on profiles $(v_i,\mathbf{s_{-i}})$ and $\mathbf{v}$ is the same. By the fact that there is no overbidding, it holds that $\mathcal{H}_{p_\textrm{A}}(v_i,\vec{s}_{-i}) \subseteq \mathcal{H}_{p_\textrm{A}}(\vec{v})$ and $\mathcal{I}_{p_\textrm{A}}(v_i,\vec{s}_{-i}) \subseteq \mathcal{I}_{p_\textrm{A}}(\vec{v})$, which means that there is more available supply for buyer $i$ on input $(v_i,\vec{s_{-i}})$. Since Mechanism $\textrm{A}$ is supply-monotone, it allocates $x_i^\textrm{A}(\vec{v})$ to buyer $i$ on profile $(v_i,\vec{s}_{-i})$.
	\medskip
	
	\noindent From the two cases above, we conclude that $x_i^\textrm{A}(v_i,\mathbf{s_{-i}}) \geq x_i^\textrm{A}(\mathbf{v})$ for the deviating buyer. By the equilibrium condition for profile $\mathbf{s}$ and for the deviation $v_i$ of buyer $i$, we have:
	\begin{eqnarray*}
		v_i\cdot x_i^\textrm{A}(\mathbf{s})-p_\textrm{A}(\mathbf{s})\cdot x_i^\textrm{A}(\mathbf{s}) &\geq& 
		v_i\cdot x_i^\textrm{A}(v_i,\mathbf{s_{-i}})-p_{\textrm{A}}(v_i,\mathbf{s_{-i}})\cdot x_i^\textrm{A}(v_i,\mathbf{s_{-i}})\\
		&\geq& v_i\cdot x_i^\textrm{A}(\mathbf{v})-p_\textrm{A}(v_i,\mathbf{s_{-i}})\cdot x_i^\textrm{A}(\mathbf{v})
		\geq v_i\cdot x_i^\textrm{A}(\mathbf{v})-p_\textrm{A}(\mathbf{v})\cdot x_i^\textrm{A}(\mathbf{v}),
	\end{eqnarray*}
	where the second inequality holds because $v_i \geq p_\textrm{A}((v_i,\mathbf{s_{-i}}))$ and $x_i^\textrm{A}(v_i,\mathbf{s_{-i}}) \geq x_i^\textrm{A}(\mathbf{v})$ and the last inequality holds because  and $p_\textrm{A}(v_i,\mathbf{s_{-i}}) \leq p_\textrm{A}(\mathbf{v})$, as explained earlier. The loss in welfare by this buyer is then bounded by $		v_i (x_i^\textrm{A}(\mathbf{v}) - x_i^\textrm{A}(\mathbf{s})) \leq p_\textrm{A}(\mathbf{v}) \cdot x_i^\textrm{A}(\mathbf{v}) - p_\textrm{A}(\mathbf{s}) \cdot x_i^\textrm{A}(\mathbf{s}) \leq p_\textrm{A}(\mathbf{v}) \cdot x_i^\textrm{A}(\mathbf{v})
	\leq B_i,$
	where the last inequality follows from the definition of the demand, since a buyer can never receive an allocation at a price that would exceed its budget. The total loss in welfare is then bounded by $\sum_{i \in \mathcal{U}_{p_{\textrm{A}}}(\vec{s})} B_i \leq \gamma_A \cdot \max_{i \in S} B_i \leq \gamma_A \cdot \max_{i \in N} B_i$ and the theorem follows.
\end{proof}

\begin{theorem}\label{thm:appwelfaremaximizing}[\ref{thm:welfaremaximizing} in main text]
	Let $\textrm{A}$ be a price-monotone welfare-maximizing mechanism. Then in any pure Nash equilibrium of $\textrm{A}$ where buyers do not overbid, the loss in social welfare (compared to the truth-telling outcome) is at most the maximum budget.
\end{theorem}
\begin{proof}
	The first part of the proof is identical to the proof of Theorem \ref{thm:welfare-extension}. By observing that welfare-maximizing mechanisms are non-wasteful and therefore supply-monotone, using exactly the same arguments as those in the first paragraph of the proof of Theorem \ref{thm:welfare-extension}, we can establish that the only loss in welfare is only due to each semi-hungry buyer $i \in \mathcal{U}_{p_{\textrm{A}}}(\vec{s})$ on $\mathbf{s}$, that receives $x_i^\textrm{A}(\mathbf{s})$ units on $\mathbf{s}$ and $x_i^\textrm{A}(\mathbf{v})$ units on $\mathbf{v}$, with $0 < x_i^\textrm{A}(\mathbf{s}) < x_i^\textrm{A}(\mathbf{v})$.
	
	Also, from the last paragraph of the proof of Theorem \ref{thm:welfare-extension}, we know that the loss in welfare is bounded by $\max_{i \in N} B_i$ as long as $|\mathcal{U}_{p_{\textrm{A}}}(\vec{s})|\leq1$. What remains is to prove is that $\mathcal{U}_{p_{\textrm{A}}}(\vec{s})$ is a singleton, i.e. that it can not be the case that there exist at least two buyers $i,j$ that appear semi-hungry at price $p_\textrm{A}(\mathbf{s})$ such that both $0 < x_i^\textrm{A}(\mathbf{s}) < x_i^\textrm{A}(\mathbf{v})$ and $0< x_j^\textrm{A}(\mathbf{s}) < x_j^\textrm{A}(\mathbf{v})$ hold. 

	Assume by contradiction $|\mathcal{U}_{p_{\textrm{A}}}(\vec{s})| >1$ and consider the deviation of a buyer $i \in \mathcal{U}_{p_{\textrm{A}}}(\vec{s})$, where the buyer reports $s_i' = s_i + \epsilon$, where $s_i + \epsilon$ is the next grid point on the output domain. We consider three cases for the price $p_\textrm{A}(s_i',\mathbf{s_{-i}})$ on profile $(s_i',\mathbf{s_{-i}})$:\\
	
	\noindent \textbf{Case 1:} $p_\textrm{A}(s_i',\mathbf{s_{-i}}) < s_i'$. In this case, buyer $i$ appears hungry on $(s_i',\mathbf{s_{-i}})$ and must receive $\min\{\floor{B_i/p_\textrm{A}(s_i',\mathbf{s_{-i}})},m\} \geq x_i^\textrm{A}(\mathbf{v})$ units at this price. Its difference in utility from the deviation is
		\begin{eqnarray*}
			 du &\geq& v_i \left(x_i^\mathbf{A}(\mathbf{v})-p_\textrm{A}(s_i',\mathbf{s_{-i}})\right)-v_i (x_i^\textrm{A}(\mathbf{s})-p_s) \\
			&=&
			v_i\left(x_i^\mathbf{A}(\mathbf{v})-x_i^\textrm{A}(\mathbf{s}) + p_\textrm{A}(s_i',\mathbf{s_{-i}})-p_s\right) \\
			&=&
			v_i \left(x_i^\textrm{A}(\mathbf{v})-x_i^\textrm{A}(\mathbf{s}) + p_\textrm{A}(s_i',\mathbf{s_{-i}})-s_i\right),
		\end{eqnarray*}
		where $du=v_i\left(x_i^\textrm{A}(s_i',\mathbf{s_{-i}})-p_\textrm{A}(s_i',\mathbf{s_{-i}})\right)-v_i\left(x_i^\textrm{A}(\mathbf{s})-p_s\right)$ and where the last equation holds by the fact that buyer $i$ appears semi-hungry at $\mathbf{s}$ and therefore $s_i = p_{\textrm{A}}(\vec{s})$.
		By the assumption that $x_i^\textrm{A}(\mathbf{s}) < x_i^\textrm{A}(\mathbf{v})$, it holds that $x_i^\textrm{A}(\mathbf{s}) \leq x_i^\textrm{A}(\mathbf{v}) +1$, since allocations are integers. Additionally, since $\epsilon < 1$, it holds that $p_\textrm{A}(s_i',\mathbf{s_{-i}})-s_i < 1$ and therefore the difference in utility from the devation $s_i'$ is strictly positive, violating the fact that $\mathbf{s}$ is a pure Nash equilibrium.\\
		
	\noindent \textbf{Case 2:} $p_\textrm{A}(s_i',\mathbf{s_{-i}})=s_i'$. In this case, buyer $i$ appears semi-hungry on $(s_i',\mathbf{s}_{-i})$ and receives $ 0 \leq x_i^\textrm{A}(s_i',\mathbf{s_{-i}}) \leq \min\{\floor{B_i/p_\textrm{A}(s_i',\mathbf{s_{-i}})},m\}$ units. Note however that there is enough supply at price $p_\textrm{A}(s_i',\mathbf{s_{-i}})$ to allocate at least $x_i^{\textrm{A}}(\mathbf{s})+1$ units to buyer $i$. This is because (i) buyer $i$ can afford $x_i^\textrm{A}(\mathbf{v})$ items at price $p_\textrm{A}(\mathbf{v})$ and therefore also at price $p_\textrm{A}(s_i',\mathbf{s_{-i}}) \leq p_\textrm{A}(\mathbf{v})$ and (ii) we can add the additional $x_j^\textrm{A}(\mathbf{s})$ units to buyer $i$'s allocation since buyer $j$ is not interested at price $p_\textrm{A}(s_i',\mathbf{s_{-i}})$ and since by assumption, $x_j^\textrm{A}(\mathbf{s}) > 0$. Since welfare-maximizing mechanisms are non-wasteful, buyer $i$ receives at least $x_i^\textrm{A}(\vec{s})+1$ units on $(s_i',\mathbf{s}_{-i})$  and a very similar argument to the previous case shows that this violates the equilibrium condition.\\

	\noindent \textbf{Case 3:} $p_\textrm{A}(s_i',\mathbf{s_{-i}})>s_i'$. This case is not possible as it would violate the welfare-maximizing nature of Mechanism $\textrm{A}$; by setting the price to $s_i'$ and allocating as much as possible to buyer $i$, one would obtain a higher welfare without violating envy-freeness.
	In each case, we obtain a contradiction, which implies that there can be at most one buyer that appears semi-hungry at price $p_\textrm{A}(\mathbf{s})$ on $\mathbf{s}$ that receives a smaller allocation than its allocation under the truth-telling profile, i.e. that $|\mathcal{U}_{p_{\textrm{A}}}(\vec{s})|\leq1$.
\end{proof}

\begin{theorem}[Revenue in any non-overbidding Nash equilibrium]\label{thm:apprevenue-extension}[\ref{thm:revenue-extension} in main text]
Let $\textrm{A}$ be a monotone mechanism that approximates the optimal revenue within a factor of $\beta$ (with $0\leq \beta \leq 1$). Then in every non-overbidding Nash equilibrium of $\textrm{A}$, the revenue is a $\left(\beta-\gamma_A \cdot \alpha\right)/2$ approximation of the optimal revenue for that instance, where $\alpha$ is the budget share of the market and $\gamma_A$ is the maximum number of buyers that receive partial allocations by $\mathrm{A}$.
\end{theorem}
\begin{proof}
	Several arguments will differ from the proof of Theorem \ref{thm:revenue-guarantee}, since we are now considering all non-overbidding equilibria and not just those that can be reached from truth-telling by best-responding. Similarly to before, for a fixed vector of budgets $\vec{B}$ and number of units $m$, let $p_\textrm{A}(\vec{\tilde{v}})$ denote the price and $x_i^\textrm{A}(\vec{\tilde{v}})$ denote the allocation buyer $i$  on input $\mathcal{M}= (\vec{\tilde{v},B},m)$.	
	
	Let $\mathbf{s}$ be a pure Nash equilibrium of Mechanism $\textrm{A}$. First, notice that by the fact that buyers do not overbid, it holds that $s_i \leq v_i$ for all $i \in N$, which in turns implies that $p_\textrm{A}(\mathbf{s}) \leq p_\textrm{A}(\mathbf{v})$, by the price-monotonicity of Mechanism $\textrm{A}$. Similarly to the proof of Theorem \ref{thm:revenue-guarantee}, we consider two cases.\\
	
	\noindent	\textbf{Case 1:} $p_\textrm{A}(\vec{s}) = p_\textrm{A}(\vec{v})$. Consider any buyer $i \in H_{p_\textrm{A}}(\vec{v})$ that receives $0\leq x_i^\textrm{A}(\vec{s})<x_i^\textrm{A}(\vec{v})$ units on $\mathbf{s}$, and consider its deviation from $\vec{s}$ to truth-telling $v_i$, with the resulting profile $(v_i,\vec{s_{-i}})$. By the price-monotonicity of Mechanism $\textrm{A}$, it holds that $p_\textrm{A}(\vec{s}) \leq p_\textrm{A}(v_i,\vec{s_{-i}}) \leq p_\textrm{A}(\mathbf{v})$, which means that $p_\textrm{A}(v_i,\vec{s_{-i}}) = p_\textrm{A}(\mathbf{s})$. However on $(v_i,\vec{s_{-i}})$, buyer $i$ now appears hungry at the same price $p_\textrm{A}(\vec{s})$ and receives $x_i^\textrm{A}(\mathbf{v})$ units, violating the fact that $\mathbf{s}$ is a pure Nash equilibrium. Therefore, every buyer $i \in H_{p_\textrm{A}}(\vec{v})$ receives $x_i^\textrm{A}(\vec{s})$ units on $\vec{s}$ and there is no loss in revenue.\\
	
	\noindent \textbf{Case 2}. $p_\textrm{A}(\vec{s}) < p_\textrm{A}(\vec{v})$. By Lemma \ref{lem:equil-one}, it holds that there do not exist buyers that receive non-zero allocations on the truth-telling profile $\mathbf{v}$ and zero allocations on the equilibrium profile $\mathbf{s}$. Therefore, on $\mathbf{s}$, each buyer receives at least as many units as they did on the true input $\vec{v}$, except possibly the set $\mathcal{U}_{p_\textrm{A}}(\mathbf{s})$ containing the semi-hungry buyers $j$ that receive allocations in $(0,\min\{\floor{B_j/p_\textrm{A}(\mathbf{s})},m\})$ (that is, each buyer $j$ in the set $\mathcal{U}_{p_\textrm{A}}(\mathbf{s})$, receives a number of units $y_j$ such that $0 < y_j < \min\left\{\floor{B_j/p_\textrm{A}(\mathbf{s})},m\right\}$).
	
	\medskip 
	
	\noindent Again, we consider the \emph{reduced market} $\mathcal{M}'=\left(\mathbf{v}_{-\mathcal{U}},\mathbf{B}_{-\mathcal{U}},m\right)$, which is obtained by the original market $\mathcal{M}=(\mathbf{v},\mathbf{B},m)$ by removing all buyers in $\mathcal{U}_{p_\textrm{A}}(\mathbf{s})$. Let $p_{min}^{-\mathcal{U}}$ denote the minimum envy-free price in $\mathcal{M}'$. We claim that $p_\textrm{A}(\mathbf{s}) \geq p_{min}^{-\mathcal{U}}$, i.e. $p_\textrm{A}(\mathbf{s})$ is an envy-free price of the reduced market. To see this, we need to establish that at price $p_\textrm{A}(\mathbf{s})$,
	\begin{enumerate}
		\item $\mathcal{I}_{p_{\textrm{A}}(\vec{s})}(\vec{s})\subseteq \mathcal{I}_{p_\textrm{A}(\vec{s})}(\vec{v})$, i.e. there are no buyers that are uninterested at price $p_\textrm{A}(\vec{s})$ on $\mathbf{v}$ and appear interested on $\mathbf{s}$.
		\item every buyer $i \in \mathcal{H}_{p_\textrm{A}}(\vec{v})\backslash \mathcal{U}_{p_\textrm{A}}(\mathbf{s})$ (i.e. who is hungry at $p_\textrm{A}(\vec{v})$) receives $x_i^\textrm{A}(\mathbf{v})$ units on $\mathbf{s}$ at price $p_\textrm{A}(\vec{s})$.
	\end{enumerate}
	The first property follows directly by the fact that buyers are not overbidding. For the second property, note that for every buyer $i \in \mathcal{H}_{p_\textrm{A}}(\vec{v})\backslash \mathcal{U}_{p_\textrm{A}}(\mathbf{s})$ that is hungry on $\mathbf{v}$ and which also appears to be hungry at $\mathbf{s}$, the property is trivially satisfied by the definition of envy-free pricing. Also, if buyer $i$ appears to be semi-hungry at $\mathbf{s}$, since $i \notin \mathcal{U}_{p_\textrm{A}}(\mathbf{s})$, it receives either $x_i^\textrm{A}(\vec{v})$ units (i.e. as many units as it can afford at this price) or $0$ units at $p_\textrm{A}(\mathbf{s})$. By Lemma \ref{lem:equil-one}, the latter is not possible, since $i \in \mathcal{H}_{p_\textrm{A}}(\vec{v})$, and therefore it follows that $x_i^\textrm{A}(\vec{s}) = x_i^\textrm{A}(\vec{v})$. This establishes the second property which in turn implies that $p_\textrm{A}(\vec{s}) \geq p_{min}^{\mathcal{-U}}$, i.e. that $p_\textrm{A}(\mathbf{s})$ is an envy-free price of the reduced market $\mathcal{M}'$.
	
	Recall the definitions of the first paragraph in the proof of Theorem \ref{thm:revenue-guarantee}. 
	By Property (2) above, it follows that $\mathcal{REV}_\textrm{A}(\mathbf{s},\mathbf{B},p_\textrm{A}(\mathbf{s}))=\mathcal{REV}_0(\mathbf{v}_{\mathcal{-U}},\mathbf{B}_{\mathcal{-U}},p_\textrm{A}(\mathbf{s}))$ and therefore it suffices to lower bound $\mathcal{REV}_0(\mathbf{v}_{\mathcal{-U}},\mathbf{B}_{\mathcal{-U}},p_\textrm{A}(\mathbf{s}))$, i.e. the minimum possible revenue attainable at price $p_\textrm{A}(\mathbf{s})$. In order to establish that $$\mathcal{REV}_0(\mathbf{v}_{\mathcal{-U}},\mathbf{B}_{\mathcal{-U}},p_\textrm{A}(\mathbf{s})) \geq (1/2)\mathcal{REV}(\mathbf{v}_{\mathcal{-U}},\mathbf{B}_{\mathcal{-U}},p_\textrm{A}(\vec{v})),$$ the arguments are identical to those used in the proof of Theorem \ref{thm:revenue-guarantee} (with the only notational difference that the index $-\ell$ is now replaced by $\mathcal{-U}$).\\
	
	\noindent Additionally, we have that
\begin{equation*}
	\mathcal{REV}(\mathbf{v}_{-\mathcal{U}},\mathbf{B}_{-\mathcal{U}},p_\textrm{A}(\vec{v})) \geq \mathcal{REV}(\mathbf{v},\mathbf{B},p_\mathrm{A}(\mathbf{v})) - \sum_{i \in \mathcal{U}_{p_\textrm{A}}(\mathbf{s})}x_i^\textrm{A}(\vec{v}) \cdot p_\mathrm{A}(\vec{v}),
\end{equation*}
	because by simply choosing price $p_\textrm{A}(\mathbf{v})$ on $\mathcal{M}'=(\mathbf{v}_{\mathcal{-U}},\mathbf{B}_{\mathcal{-U}},m)$ (which is an envy-free price), we lose at most the contribution to the revenue of the buyers in $\mathcal{U}_{p_\textrm{A}}(\mathbf{s})$. Using a very similar calculation as in the proof of Theorem \ref{thm:revenue-guarantee}, we obtain that
	\begin{eqnarray*}
		\mathcal{REV}_{\textrm{A}}(\vec{s},\vec{B}) &\geq& \frac{1}{2}\cdot  \mathcal{REV}(\mathbf{v}_{\mathcal{-U}},\mathbf{B}_{\mathcal{-U}},p_\textrm{A}(\vec{v}))\\
		&\geq& \frac{\beta}{2} \cdot \mathcal{REV}(\vec{v},\vec{B}) - \frac{\sum_{i \in \mathcal{U}_{p_\textrm{A}}(\mathbf{s})}B_i}{2} \\
		&\geq& \frac{1}{2}\left(\beta-\gamma_A \cdot \alpha\right)\mathcal{REV}(\vec{v},\vec{B})
	\end{eqnarray*}	
	where the last inequality follows from the definition of the budget share. This completes the argument.
\end{proof}

\section{The \textsc{All-or-Nothing} mechanism}\label{app:section5}

To give some intuition, we start with a few simple examples of equilibria of $\textsc{All-Or-Nothing}$ which are not the dominant-strategy, truth-telling equilibria. Note that the set of equilibria the we present contain both overbidding and non-overbidding equilibria and equilibria where the price is either higher, lower or equal to the price under truth-telling.
\begin{example}
	Consider the following profile of true valuations with $m=10$ units and $n=5$ buyers such that for each buyer $i \in \{1,2,3,4,5\}$, the buyer has a valuation-budget pair $(v_i,B_i)$ given by the following:
	\begin{eqnarray*}
	(v_1,B_1) = (2,2) ,\ \ (v_2,B_2) = (2,2), \ \ (v_3,B_3) = (1,6), \ \ 
	 (v_4,B_4) = (0.5,1), \ \ (v_5,B_5) = (0.5,2) 
	\end{eqnarray*}
	Consider the following profiles, which are easily verifable to be pure Nash equilibria:
	\begin{enumerate}
		\item All buyers are truth-telling. In that case, the price is $p=1$, buyers $1$ and $2$ receive 2 units each and all other buyers receive $0$ units. This is the dominant-strategy equilibrium.
		\item $v_1' = 3$, $v_5' = 0.4$ and $v_i' = v_i$ for $i \in \{2,3,4\}$. The price is again $p=1$, buyers $1$ and $2$ receive 2 units each and all other buyers receive $0$ units.
		\item $v_4' = v_5'= 0.1$, $v_i' = v_i$ for $i \in \{1,2,3\}$. The price is again $p=1$, buyers $1$ and $2$ receive 2 units each and all other buyers receive $0$ units.
		\item $v_3' = 0.5$, $v_i' = v_i$ for $i \in \{1,2,4,5\}$. The price is now $p=0.5$, buyers $1$ and $2$ receive $4$ units each, buyer $4$ receives the remaining $2$ units and buyers $3$ and $4$ receive $0$ units.
	\end{enumerate}
\end{example}

We will show that for $n=2$ buyers, the best response dynamic of this mechanism converges no matter what the starting profile is.
\begin{theorem} \label{thm:appaon_converge}[\ref{thm:aon_converge} in main text]
For $n=2$ buyers, the best response dynamic of the \textsc{All-or-Nothing} mechanism converges to a Nash equilibrium from any initial strategy profile.
\end{theorem}

\noindent In order to prove this theorem we develop first a series of lemmas. We start with a simple lemma with some immediate properties of the mechanism.
\begin{lemma}\label{observation1}
	The following facts hold about the \textsc{All-Or-Nothing} mechanism:
	
	\begin{enumerate}
		\item \textsc{All-Or-Nothing} is monotonic.\medskip
		\item Given a price $p$, \textsc{All-Or-Nothing} allocates to each buyer either {\em all} or {\em nothing} at this price, but never anything in between.\medskip
		\item In \textsc{All-Or-Nothing}, a buyer will never best-repond twice in a row.
	\end{enumerate}
\end{lemma}
\begin{proof}
	We argue each of the items (1)-(3) above individually:
\begin{enumerate}
	\item First, the mechanism is price-monotonic as it always outputs the minimum envy-free price. If a buyer reports a lower value, the minimum envy-free price can not increase and similarly, if a buyer increases its reported value, the minimum envy-free price can not decrease. For supply-monotonicity, note that since the mechanism allocates buyers to the semi-hungry buyers in lexicographic order, it is not possible for a buyer that receives fewer units at the same price, when the number of interested buyers is decreased. From these two properties, we obtain that \textsc{All-Or-Nothing} is monotonic.\medskip
	\item This is by the definition of the mechanism. If a buyer is hungry, it receives all the units in its demand, i.e. \emph{all}. If a buyer is semi-hungry, it either receives all possible units that it can afford, i.e. \emph{all} or no units at all, i.e. \emph{nothing}.\medskip
	\item Assume by contradiction that some agent could best-respond twice in a row from profile $\vec{s}_1$ to $\vec{s}_2$ and then to $\vec{s}_3$. By the definition of \textsc{All-Or-Nothing}, the outcome of the mechanism would be the same as the one obtained if the agent had best-responded from $\vec{s}_1$ to $\vec{s}_3$, contradicting the fact that the move from $\vec{s}_1$ to $\vec{s}_2$ was a best-response.
\end{enumerate}
\end{proof}

\noindent The following definition will be useful.
\begin{definition}[Type $-,0,+$]
	Let $\vec{s}=(s_1,\ldots,s_n)$ be any strategy profile. Then a buyer $i$ has type:
	\begin{itemize}
		\item $+$: if its utility  in $\vec{s}$ is positive.
		\item $0$: if its utility in $\vec{s}$ is $0$.
		\item $-$: if its utility in $\vec{s}$ is negative.
	\end{itemize}
	We will say that an agent is a ``type h'' agent, for $h \in \{-,0,+\}$. Furthermore, given a profile $\vec{s}=(s_1,\ldots,s_n)$ we will use a vector $(h_1,\ldots,h_n)$ to denote the buyers' types.
\end{definition}

\noindent We have the following observation regarding buyers of different types.

\begin{observation}\label{remark1}
\emph{	If a buyer $i$ is of type:}
	\begin{itemize}
		\item \emph{$+$: then its true value $v_i$ is strictly larger than the price $p_{\vec{s}}$ on profile $\mathbf{s}$, its reported value $s_i$ is not smaller than the price $p_\vec{s}$, and the buyer is given all the units it can afford by the mechanism, i.e. \emph{all}.}\medskip
		\item \emph{$-$: then its true value $v_i$ is strictly smaller than the price $p_{\vec{s}}$ on profile $\vec{s}$, its reported value $s_i$ is not smaller than the price $p_\vec{s}$, the buyer is given all the units it can afford by the mechanism, i.e. \emph{all} and the buyer can afford at least one unit. We will refer to such buyers as \emph{wrongfully interested}.}\medskip
		\item \emph{$0$: then there are three cases:}\medskip
		 \begin{enumerate}
			\item \emph{the buyer appears irrelevant (i.e. its reported value $s_i$ is larger than the price $p_\vec{s}$ on profile $\mathbf{s}$, but it receives $0$ units at $p_\vec{s}$ due to insufficient budget).}\medskip
			\item \emph{the buyer's reported value $s_i$ is not larger than $p_\vec{s}$ and its true value $v_i$ is not larger that $p_\vec{s}$. If $s_i = p_\vec{s}$, then the buyer receives $0$ units due to tie-breaking for the semi-hungry buyers.}\medskip
			\item \emph{the buyer's reported value $s_i$ is not larger than $p_\vec{s}$ and its true value $v_i$ is larger that $p_\vec{s}$. If $s_i = p_\vec{s}$, then the buyer receives $0$ units due to tie-breaking for the semi-hungry buyers.}
		\end{enumerate}
	\end{itemize}
\end{observation}

\noindent We will next prove a series of lemmas, which will be used for proving Theorem \ref{thm:aon_converge}. Note that in the following, Lemmas \ref{fact:nominprice}, \ref{fact:noplusdeviation}, \ref{fact:minustozero} and \ref{lemma:viable} hold for any number of buyers.

\begin{lemma}\label{fact:nominprice}
	A type $+$ or type $0$ buyer can not best-respond to increase its utility by lowering the price.
\end{lemma}	
\begin{proof}
	Let $p$ be the price before the deviation and assume by contradiction that such a buyer could increase is utility by lowering the price. This implies that at the new price $q$ \textsc{All-Or-Nothing} after the deviation, the buyer receives as many units as it can afford, i.e. {\em all}. Since all the other reports are fixed, this price would be an envy-free price on the original profile, before the deviation, contradicting the minimality of $p$.
\end{proof}

\begin{lemma}\label{fact:noplusdeviation}
	Any type $+$ buyer is already best-responding.
\end{lemma}

\begin{proof}
	Let $p$ be the price before the deviation. Since the buyer was getting as many units as it could afford, i.e. {\em all} at price $p$, it must be getting all the units that it can afford at a lower price $q$, outputted by \textsc{All-Or-Nothing} after the deviation, for its utility to be higher. This is not possible, by Lemma \ref{fact:nominprice}.
\end{proof}

\begin{lemma}\label{fact:minustozero}
	A type $-$ buyer can only best-respond to become a type $0$ buyer.
\end{lemma}

\begin{proof}
	First note that a switch from a $-$ type to a $-$ type (where the utility is still negative but higher) is not a best-response, since a buyer can always switch to reporting $0$. Since $0$ is not an envy-free price, the mechanism will never output it and therefore the buyer will never be allocated any items at this price, for a utiltiy of $0$. 
	
	Suppose now that the buyer switches from a $-$ type to a $+$ type and let $p$ be the price before the deviation and $q$ be the price after the deviation. For buyer i to be a $-$ type buyer, by Observation \ref{remark1}, it has to be the case that (a) its true value is smaller than $p$ and (b) its reported value is weakly larger than $p$ and it receives \emph{all} at price $p$, and $\emph{all}$ is at least one unit. For the buyer to be a $+$ type buyer, it has to receive {\em all} at $q$. By the same argument as the one used in the proof of Lemma \ref{fact:nominprice}, since all the other buyers are fixed, $q$ would be an envy-free price in the original profile (before the deviation), contradicting the minimality of $p$.
\end{proof}

\begin{lemma}\label{lemma:viable}
	The only viable deviations are 
	\begin{enumerate}
		\item from type $-$ to type $0$. Any such deviation does not increase the price.
		\item from type $0$ to type $+$. Any such deviation does not decrease the price.
	\end{enumerate}
\end{lemma}
\begin{proof}
	The fact that the only viable deviations are from $-$ to $+$ or from $0$ to $+$ follow from Lemma \ref{fact:noplusdeviation} and Lemma \ref{fact:minustozero}. In order to argue the changes in price, let $p$ be the price before the deviation and $q$ be the price after the deviation. Consider first a deviation from type $-$ to type $0$ and assume by contradiction that it increases the price, i.e. $q < p$. By Lemma \ref{remark1}, type $-$ receives {\em all} at price $p$ (and \emph{all} is a non-zero allocation), but its true valuation is lower than $p$. By assumption, $q > p$ and therefore by the monotonicity of \textsc{All-Or-Nothing} (Lemma \ref{observation1}), it must be the case that buyer i increased its report when deviating. Since everything else is fixed and the buyer received {\em all} before, it still receives {\em all} after deviating. Since the sets of hungry buyers and semi-hungry buyers  at price $p$ a has not changed after the deviation, $p$ is still an envy-free price on the new profile obtained after the deviation, contradicting the minimality of $q$.
	
	Next, consider a deviation from a $0$ type to a $+$ type. For the buyer to be a $0$ type buyer, there are three possibilities, presented in Observation \ref{remark1} above. 
	In cases (1) and (2) of the observation, the buyer can only increase its utility by lowering the price which is not possible by Lemma \ref{fact:nominprice}, therefore these buyers are already best-responding. In case (3), the buyer again can not lower its utility by lowering the price by Lemma \ref{fact:nominprice}, therefore the new price $q$ has to be at least as large as $p$. 
\end{proof}

Given Lemma \ref{lemma:viable}, it suffices to consider deviations of buyers from type $-$ to type $0$ and from type $0$ to type $+$. 

\begin{lemma} \label{lem:minplus}
	From a $(-,+)$ or $(+, -)$ state an equilibrium is reached in at most one round.
\end{lemma}
\begin{proof}
	By Lemma \ref{fact:noplusdeviation}, a $+$ buyer is already best-responding, so any improving deviation will be carried out by the $-$ buyer. W.l.o.g. let Alice be the $-$ buyer (the argument for the $(+,-)$ state will follow by symmetry), $p$ be the current price, and $q$ the price after Alice's best response. By Lemma \ref{lemma:viable}, the price weakly decreases when a $-$ buyer best responds, so $q \leq p$. There are three cases. If the new state is $(0,+)$, Alice just deviated so by Lemma \ref{observation1}, she is now best-responding. Bob is a type $+$ buyer, so by Lemma \ref{fact:noplusdeviation} he is also best-responding and the new profile is an equilibrium. If the new state is $(0,0)$, then we observe that this case is in fact not possible, since $q \leq p$, Bob received all the units he could afford at $p$, and his report did not change and the price did not increase, so he must still receive all the units he can afford at $q$, which implies his new utility cannot be zero. If the new state is $(0,-)$ we obtain the same contradiction in Bob's utility, so this cannot state is not reachable either. 
	Thus we obtain that the $(-,+)$ state is either an equilibrium to begin with or results in an equilibrium in one round.
\end{proof}

\begin{lemma}
	In $(0,+), (+,0) , \ldots, (+,0), (0,+), \ldots$ sequence, if both buyers best-respond consecutively without changing the price, then we are at an equilibrium.
\end{lemma}	

\begin{proof}
Consider the point in the sequence where the first best-response that does not change the price occurs and let $p$ be that price. Also, assume without loss of generality that Alice is best-responding next, hence we are at a $(0,+)$ type profile. 

First note that Bob appears (before and after Alice's deviation) semi-hungry at price $p$. This is true because if Bob was \begin{itemize}
	\item[-] hungry at price $p$ (and since he was not irrelevant, as he was a type $+$ buyer) and the price did not change, he would still be a $+$ type buyer after Alice's deviation. 
	\item[-] uninterested at price $p$, he wouldn't be a type $+$ buyer before Alice's deviation.
\end{itemize}
Additionally, 
\begin{itemize}
	\item[-] \emph{before Alice's deviation}, Bob received as many units as he could afford (i.e. \emph{all}, with \emph{all} being non-zero) at price $p$. This follows from the fact that Bob was a type $+$ buyer before Alice's deviation.
	\item[-] \emph{after Alice's deviation}, Bob receives $0$ units at price $p$. This follows from the fact that Bob is a type $0$ buyer after Alice's deviation and the price has not changed.
\end{itemize}
Now consider Alice after her deviation. Since she is a type $+$ buyer, she either
\begin{itemize}
	\item [-] appears hungry at price $p$ or
	\item [-] appears semi-hungry at price $p$ and receives as many units as she can afford at $p$, i.e. \emph{all} by the mechanism, where \emph{all} is a non-zero quantity.
\end{itemize}
Finally, consider Bob's best-response (to the $(+,0)$ profile next in the sequence), which by assumption, does not change the price $p$. By the discussion above regarding Alice and Bob after Alice's deviation, it follows that Bob can only become a type $+$ buyer if he appears hungry after his best-response, which is only possible if Alice was semi-hungry before Bob's best-response (as otherwise his best-response would result in a  $(+,+)$ profile, contradicting the nature of the sequence). Regardless of what Alice reports next, Bob can not become a type $0$ buyer again, if the price remains the same. Therefore, since by assumption the next state is $(0,+)$, we are at an equilibrium.
\end{proof}

\begin{lemma} \label{lem:zeroplus}
	From a $(0,+)$ or $(+,0)$ state an equilibrium is always reached.
\end{lemma}
\begin{proof}
	The $+$ buyer is already best responding, so an improving deviation can only be obtained by the $0$ buyer.
	W.l.o.g. let Alice be the $0$ buyer, $p$ be the price at the initial $(0,+)$ state, and $q$ the price after her deviation. By Lemma \ref{lemma:viable}, $q \geq p$. The change of state after Alice's deviation is given by one of the following:
	\begin{description}
		\item[1.] $(0,+) \rightarrow (+,-)$: Bob received all the units he could afford before and still receives all the units he can afford in the new state (at a possibly higher price $q$) but is not wrongfully interested. By Lemma \ref{lemma:viable}, he best responds with becoming a type $0$ buyer and new chosen price $q'$ is not larger than $q$. Therefore, since Bob just best-responded and Alice is still a $+$ type at the new price, by Lemma \ref{fact:noplusdeviation}, \emph{we are at an equilibrium}.\medskip
		
		\item[2.] $(0,+) \rightarrow (+,+)$: By Lemma \ref{fact:noplusdeviation}, both buyers are already best responding, so this is an equilibrium.\medskip
		
		\item[3.] $(0,+) \rightarrow (+,0)$: If Bob has no improving deviation, we reached an equilibrium. Otherwise, after the deviation Bob will become a type $+$, yielding one of the states:\medskip
		\begin{description}
			\item[a.] $(+,+)$: An equilibrium by Lemma \ref{fact:noplusdeviation}.\medskip
			\item[b.] $(-,+)$: Leads to an equilibrium by Lemma \ref{lem:minplus}.\medskip
			\item[c.] $(0,+)$: This is the case of interest, where the price increased as a result of Bob's best-response and Bob became $+$ while turning Alice into $0$. If Alice has any best-response at this point, she can only become a $+$ as a result of it while turning Bob into one of $-,0,+$. If Bob becomes $-$ or $+$, we fall into cases $3.b$ and $3.a$, respectively, which lead to an equilibrium.\medskip
			
			Thus to complete the argument we must show that an alternating sequence of the form $(0,+), (+,0) , \ldots, (+,0), (0,+), \ldots$ will converge to an equilibrium. By Lemma \ref{lemma:viable}, the price cannot decrease along such a sequence. If both buyers best-respond along the sequence without changing the price, then we are at an equilibrium by Lemma. Otherwise, if the price is increased in every round along the sequence, the process must stop before the price is higher than the sum of budgets (at such a price no buyer can afford anything, and at least one of them had strictly positive utility before).
		\end{description}
	\end{description}
\end{proof}
\begin{lemma} \label{lem:zerozero}
	From a $(0,0)$ state an equilibrium is always reached.
\end{lemma}
\begin{proof}
	W.l.o.g. Alice is best responding from this state. If the new state is $(+,-)$, we will reach an equilibrium by Lemma \ref{lem:minplus}. If the new state is $(+,+)$ then both buyers are already best responding. (recall that by Lemma \ref{fact:noplusdeviation}, a $+$ buyer is already best responding.). If the new state is $(+,0)$, then after Alice's deviation, Bob is a $0$ type and by Lemma \ref{lemma:viable}, he will best-respond (if he can) to become a $+$ type. If he does not change Alice's type, we have a $(+,+)$ profile which by Lemma \ref{fact:noplusdeviation}, is an equilibrium. The two remaining cases for the profile obtained after Bob's deviation are $(-,+)$, which leads to an equilibrium by Lemma \ref{lem:minplus} and $(0,+)$, which leads to an equilibrium by Lemma \ref{lem:zeroplus}.
	Thus an equilibrium is always reached as required.
\end{proof}

\begin{lemma} \label{lem:minzero_min}
	From a $(-,0)$ or $(0,-)$ state an equilibrium is always reached if the $-$ buyer is best responding next.
\end{lemma}
\begin{proof}
	W.l.o.g. Alice is the $-$ buyer and Bob the $0$ buyer; the $(0,-)$ case holds by symmetry. By the condition in the lemma, Alice is best responding next. The change of state after her deviation is given by one of the following:
	\begin{description}
		\item[1.] $(-,0) \rightarrow (0,+)$: Alice just deviated so she is best responding. Bob is $+$, so he is already best responding by Lemma \ref{fact:noplusdeviation} and the new profile is an equilibrium.\medskip
		
		\item[2.] $(-,0) \rightarrow (0,0)$: By Lemma \ref{lem:zeroplus}, this state always leads to an equilibrium.\medskip
		
		\item[3.] $(-,0) \rightarrow (0,-)$: Since Alice just moved from $-$ to $0$, it's Bob's turn to deviate. Bob's type is $-$, so after his deviation he will become $0$ and the price will weakly decrease. We have several subcases depending on Alice's type after Bob's deviation. If the new state is: \medskip
		\begin{description}
		\item[a.] $(0,0)$: by Lemma \ref{lem:zeroplus} we reach an equilibrium. \medskip
		\item[b.] $(+,0)$: then we are at an equilibrium since Bob just best responded and Alice is already best responding since she is a $+$ type. \medskip
		\item[c.] $(-,0)$: then if we do not enter one of the states above (a) or (b) at all, there exists an alternating sequence of deviations where each state in the sequence is $(-,0)$ or $(0,-)$. Since the initial state is such that the $-$ buyer is best responding, the buyer with $-$ type will always be the one deviating along such a sequence. First, by Lemma \ref{lemma:viable}, the price weakly decreases with every such deviation, since the deviator always moves from $-$ to $0$. Second, whenever the price stays constant in such a deviation, the deviator must decrease his reported utility. Thus along every such sequence either the price decreases or the reported valuation of the deviator. This means that the process stops in finite time.
		\end{description}
	\end{description}
	Thus we reach an equilibrium in all cases as required.
\end{proof}

Combining Lemmas \ref{lem:minzero_min} and \ref{lem:minzero_zero}, we obtain that an equilibrium is always reached from a $(-,0)$ or $(0,-)$ state.
\begin{lemma} \label{lem:minzero}
	From a $(-,0)$ or $(0,-)$ state an equilibrium is always reached.
\end{lemma}
\begin{proof}
	This is an immediate corollary of Lemmas \ref{lem:minzero_min} and \ref{lem:minzero_zero}.
\end{proof}

\begin{lemma} \label{lem:minzero_zero}
	From a $(-,0)$ or $(0,-)$ state an equilibrium is always reached if the $0$ buyer is best responding next.
\end{lemma}
\begin{proof}
	W.l.o.g. Alice is the $0$ buyer, so the state is $(0,-)$. Let $p$ be the current price and $q$ the price after her deviation. By Lemma \ref{lemma:viable}, $q \geq p$. If the new state is $(+,-)$, we reach an equilibrium by Lemma \ref{lem:minplus}. If the new state is $(+,0)$, we reach an equilibrium by Lemma \ref{lem:zeroplus}. If the new state is $(+,+)$: By Lemma \ref{fact:noplusdeviation}, at $(+,+)$ both buyers are already best-responding. 
\end{proof}

\begin{lemma} \label{lem:minmin}
	From a $(-,-)$ state an equilibrium is always reached. 
\end{lemma}
\begin{proof}
	W.l.o.g., Alice is deviating next. By Lemma \ref{fact:minustozero}, she will be become a $0$ type. If the new state is $(0,+)$, then since Alice just deviated she is already best-responding, and Bob is also best responding by Lemma \ref{fact:noplusdeviation}, so the new profile is an equilibrium. If the new state is $(0,0)$, we reach an equilibrium by Lemma \ref{lem:zerozero}, and if it is $(0,-)$, then we reach an equilibrium by Lemma \ref{lem:minzero}.
\end{proof}

We can now prove the convergence theorem.

\begin{proof}[Proof of Theorem \ref{thm:appaon_converge}][\ref{thm:aon_converge} in main text]
	Let $\mathcal{M} = (m,\vec{B}, \vec{v})$ be a market with two buyers, Alice and Bob.
	Consider any strategy profile $\vec{s} = (s_1, s_2)$.
	By Lemma \ref{fact:noplusdeviation}, only buyers with strictly negative or zero utility at $\vec{s}$ can have improving deviations from $\vec{s}$. Thus if $\vec{s}$ is a $(+,+)$ state, it is already an equilibrium. Otherwise, by Lemmas \ref{lem:minmin}, \ref{lem:minzero}, \ref{lem:minplus}, \ref{lem:zerozero}, \ref{lem:minzero}, the best response sequence stops when initiated from any of the states $(-,-)$, $(-,0)$, $(0,-)$, $(-,+)$, $(+,-)$, $(0,0)$, $(0,+)$, $(+,0)$.
\end{proof}

\end{document}